\documentclass[11pt,english]{article}
\usepackage{geometry}
\usepackage{setspace}
\usepackage[T1]{fontenc}
\usepackage[utf8]{inputenc}
\usepackage{babel}
\usepackage{varioref}
\usepackage{amsmath}
\numberwithin{equation}{section}
\usepackage{amssymb}
\usepackage{graphicx}
\usepackage{floatrow}
\usepackage[labelfont=bf]{caption}
\usepackage{enumitem}
\usepackage[toc,page]{appendix}
\newcommand{\nocontentsline}[3]{}
\newcommand{\tocless}[2]{\bgroup\let\addcontentsline=\nocontentsline#1{#2}\egroup}
\newcommand{\stoptocwriting}{%
  \addtocontents{toc}{\protect\setcounter{tocdepth}{-5}}}
\newcommand{\resumetocwriting}{%
  \addtocontents{toc}{\protect\setcounter{tocdepth}{\arabic{tocdepth}}}}
\usepackage{slashed}
\usepackage{svg}
\usepackage [autostyle, english = american]{csquotes}
\MakeOuterQuote{"}
\makeatletter

\usepackage{xcolor}
\usepackage{authblk}

\usepackage{amsthm}
\usepackage{graphicx}

\usepackage{mathtools}
\usepackage{comment}

\usepackage{hyperref}
\hypersetup{
    colorlinks,
    linkcolor={blue!60!black},
    citecolor={green!60!black},
    urlcolor={blue!80!black}
}
\newtheorem{thm}{Theorem}[section]

\newtheorem{defi}{Definition}[section]

\newtheorem{cor}{Corollary}[section]
\theoremstyle{remark}
\newtheorem{rem}{Remark}[section]
\theoremstyle{plain}
\newtheorem{prop}{Proposition}[section]

\newenvironment{nalign}{
    \begin{equation}
    \begin{aligned}
}{
    \end{aligned}
    \end{equation}
    \ignorespacesafterend
}
\theoremstyle{remark}
\usepackage{todonotes}

\newcommand{\dd}{\mathop{}\!\mathrm{d}}
\newcommand{\pu}{\partial_u}
\newcommand{\pv}{\partial_v}
\newcommand{\A}{\mathcal{A}}

\newcommand{\Hp}{\mathcal{H}^+}

\newcommand{\con}{\text{const.}}
\newcommand{\ill}{I_0^{\log}}
\newcommand{\ilo}{I_0^{\frac{\log r}{r^2}}}
\newcommand{\ilog}{I_0^{\frac{\log r}{r^3}}}

\title{The Case Against Smooth Null Infinity II:\\A Logarithmically Modified Price's Law} 

\author[1]{Leonhard M. A. Kehrberger} 
\affil[1]{University of Cambridge, Department of Applied Mathematics and Theoretical Physics, Wilberforce Road, Cambridge CB3 0WA, United Kingdom}
\date{\today} 

\makeatother

\begin{document}
\pagenumbering{gobble}

\maketitle 
\begin{abstract}
In this paper, we expand on results from our previous paper "The Case Against Smooth Null Infinity I: Heuristics and Counter-Examples"~\cite{Kerrburger} by showing that the failure of "peeling" (and, thus, of smooth null infinity) in a neighbourhood of $i^0$ derived therein translates into logarithmic corrections at leading order to the well-known Price's law asymptotics near $i^+$. This suggests that 
 \textit{the non-smoothness of $\mathcal{I}^+$ is physically measurable}.

More precisely, we consider the linear wave equation $\Box_g \phi=0$ on a fixed Schwarzschild background ($M> 0$), and we show the following: 
If one imposes conformally smooth initial data on an ingoing null hypersurface (extending to $\mathcal{H}^+$ and terminating at $\mathcal{I}^-$) and vanishing data on $\mathcal{I}^-$ (this is the no incoming radiation condition), then the precise leading-order asymptotics of the solution $\phi$ are given by $r\phi|_{\mathcal{I}^+}=C u^{-2}\log u+\mathcal{O}(u^{-2})$ along future null infinity, $\phi|_{r=R>2M}=2C\tau^{-3}\log\tau+\mathcal{O}(\tau^{-3})$ along hypersurfaces of constant $r$, and $\phi|_{\mathcal{H}^+}=2Cv^{-3}\log v+\mathcal{O}(v^{-3})$ along the event horizon.
Moreover, the constant $C$ is given by $C=4M I_0^{(\mathrm{past})}[\phi]$, where  $I_0^{(\mathrm{past})}[\phi]:=\lim_{u\to -\infty} r^2\partial_u(r\phi_{\ell=0})$ is the \textit{past Newman--Penrose constant} of $\phi$ on $\mathcal{I}^-$. 

Thus, the precise late-time asymptotics of $\phi$ are completely determined by the early-time behaviour of the spherically symmetric part of $\phi$ near $\mathcal{I}^-$. 

Similar results are obtained for polynomially decaying timelike boundary data. 

The paper uses methods developed by Angelopoulos--Aretakis--Gajic and is essentially self-contained.

\end{abstract}
\newpage
\begingroup
\hypersetup{linkcolor=black}
    \tableofcontents{}
    \endgroup
\newpage
\pagenumbering{arabic}
\section{Introduction}
This paper is concerned with the study of the precise late-time asymptotics of solutions to the wave equation
\begin{equation}\label{waveequation}
\Box_g\phi=0
\end{equation}
on the exterior of a fixed Schwarzschild (or a more general, spherically symmetric) background $(\mathcal M_M,g_M)$ with mass $M$ under physically motivated assumptions on data. The most important of these assumptions is the \textit{no incoming radiation condition} on $\mathcal{I}^-$, stating that the flux of the radiation field on past null infinity vanishes at late advanced times. 

We initiated the study of such data in~\cite{Kerrburger}, where we constructed two classes of  solutions\footnote{In fact, we also constructed solutions to the non-linear Einstein-Scalar field system in~\cite{Kerrburger}.} satisfying the no incoming radiation condition (as a condition on data on $\mathcal{I}^-$). 
The first class had polynomially decaying boundary data on a timelike boundary $\Gamma$ terminating at $i^-$, whereas the second class had polynomially decaying characteristic initial data on an ingoing null hypersurface $\mathcal{C}_{\mathrm{in}}$ terminating at $\mathcal{I}^-$. 
The choice for these data  was in turn motivated by an argument due to D.\ Christodoulou~\cite{CHRISTODOULOU2002}, which showed that the assumption of \textit{Sachs peeling} and, thus, of (conformally) smooth null infinity, is incompatible with the no incoming radiation condition and the prediction of the quadrupole formula for $N$ infalling masses from $i^-$.
 Indeed, we proved that the solutions from~\cite{Kerrburger} described above are not only in agreement with the quadrupole formula (which predicts that $\pu(r\phi)\sim|u|^{-2}$ near $i^0$), but also  lead to logarithmic terms in the asymptotic expansion of $\pv(r\phi)$ as $\mathcal{I}^+$ is approached, thus contradicting the statement of Sachs peeling that such expansions can be expanded in powers of $1/r$.
Roughly speaking, we obtained for the spherically symmetric mode $\phi_0$ that if the limit
\begin{equation}\label{1.2}
\lim_{\mathcal{C}_{\mathrm{in}},u\to-\infty}|u|r\phi_0:=\Phi^-
\end{equation}
on initial data is non-zero (or if a similar condition on $\Gamma$ holds), then, for sufficiently large negative values of $u$, one obtains on each outgoing null hypersurface of constant $u$ the asymptotic expansion
\begin{equation}\label{nt:eq:intro:pvrphilog}
\pv(r\phi_0)(u,v)=-2M\Phi^- \frac{\log r-\log|u|}{r^3}+\mathcal{O}(r^{-3}).
\end{equation}

On the other hand, we will show in upcoming work~\cite{Kerrburger3} that higher $\ell$-modes, under similar assumptions, decay slower, $\pv(r\phi_\ell)\sim r^{-2}$, with logarithmic terms appearing at order $r^{-3}\log r$ (or at a later order, depending on the setting).

The above results give a complete picture for the situation near $\mathcal I^-$. 
Naturally, one may then ask how the \textit{early-time asymptotics} \eqref{nt:eq:intro:pvrphilog} translate into \textit{late-time asymptotics} when one smoothly extends\footnote{It turns out that the leading-order asymptotics will not depend on the extension.} the data on $\mathcal{C}_{\mathrm{in}}$ (or $\Gamma$)  all the way to the event horizon $\mathcal{H}^+$. 
In this work, we shall provide a detailed answer to this question. 
Let us already paraphrase the main statement (see also Figure~\ref{fig:1}):
\par\smallskip\noindent
\centerline{\begin{minipage}{0.9\textwidth}
\textit{Consider solutions $\phi$ to \eqref{waveequation} which arise from the no incoming radiation condition and from smooth data on $\mathcal C_{\mathrm{in}}$ that satisfy \eqref{1.2}. Then their leading-order asymptotic behaviour towards $i^+$ is determined by the spherical mean $\phi_0$ and contains logarithmic terms. Thus, the non-smoothness of null infinity near $i^0$ propagates and translates into logarithmic tails near $i^+$.}
\end{minipage}}
\par\smallskip
\begin{figure}[htbp]
  \floatbox[{\capbeside\thisfloatsetup{capbesideposition={right,top},capbesidewidth=4.3cm}}]{figure}[\FBwidth]
{\caption{Schematic depiction of the results of the present paper. We can either impose data $r\phi|_{\Gamma}=C/t+\dots$, which will, by our previous results \cite{Kerrburger}, lead to behaviour on $\mathcal C_{\mathrm{in}}$ as depicted, or we can directly impose data on $\mathcal C_{\mathrm{in}}$. In both cases, we obtain logarithmic late-time asymptotics near $i^+$ provided the data are extended to $\mathcal H^+$.}\label{fig:1}}
  {\includegraphics[width = 185pt]{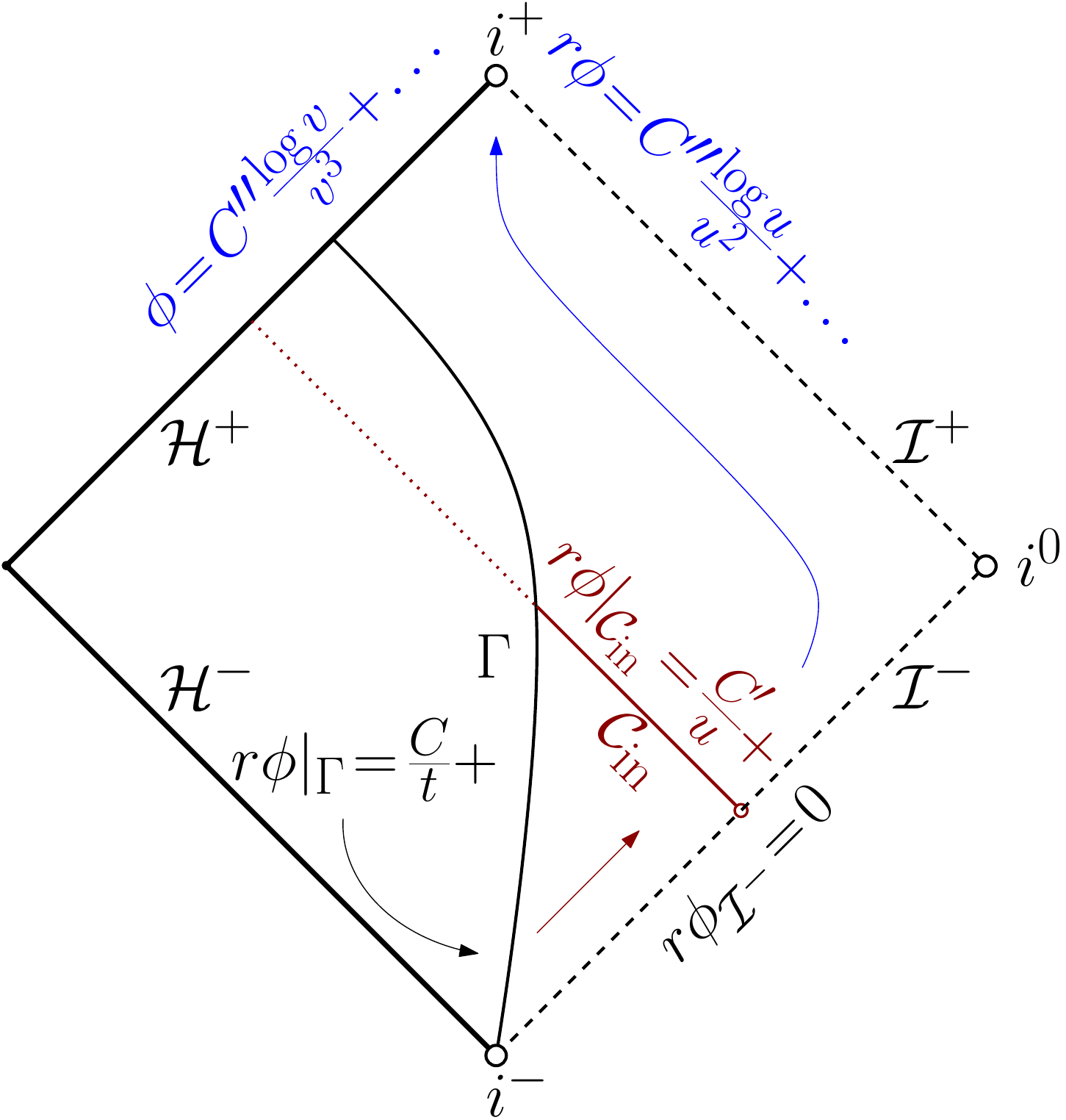}}
\end{figure}
\subsection{The relation to Price's law}
Showing a statement like the above is closely related to the task of proving Price's law. We recall that Price's law~\cite{Price72,PriceGundlachPullin94} roughly states that the evolutions of compactly supported Cauchy data under \eqref{waveequation} satisfy the following asymptotics:
$\phi|_{\mathcal{H}^+}\sim v^{-3}$ along the event horizon, $\phi|_{r=\con}\sim t^{-3}$ along hypersurfaces of constant $r$, and $r\phi|_{\mathcal{I}^+}\sim u^{-2}$ along future null infinity.

A rigorous proof of these asymptotics has only recently been obtained by Angelopoulos, Aretakis and Gajic~\cite{Angelopoulos2018ASpacetimes,Angelopoulos2018Late-timeSpacetimes} (see also the works~\cite{Hintz},~\cite{Ma2021} and~\cite{Donninger11,Donninger12}, as well as~\cite{Angelopoulos2019LogarithmicInfinity} for refined asymptotics).
 They, in fact, show that the leading-order asymptotics are determined by the $\ell=0$-mode, as higher $\ell$-modes decay at least half a power faster. 
 Their proof is split up into two parts. In the first~\cite{Angelopoulos2018ASpacetimes}, they derive \textit{almost-sharp} decay estimates with an $\epsilon$-loss, based on an extension of the $r^p$-method introduced in~\cite{Dafermos2010ASpacetimes}. 
 In the second part~\cite{Angelopoulos2018Late-timeSpacetimes}, they then use these almost-sharp estimates, together with certain conservation laws along null infinity $\mathcal I^+$ and a clever splitting into different spacetime regions, to obtain the precise leading-order asymptotics of the spherical mean. The upshot of this is that all the results in the first part~\cite{Angelopoulos2018ASpacetimes} are, in some sense, blind to logarithmic corrections; the $\epsilon$-loss in the almost-sharp decay estimates "swallows" the $\log$-terms. 
Therefore, in order to find the late-time asymptotics of solutions coming from initial data satisfying \eqref{nt:eq:intro:pvrphilog}, we only need to suitably adapt the second part of their proof~\cite{Angelopoulos2018Late-timeSpacetimes} and can use the results of~\cite{Angelopoulos2018ASpacetimes} as black box results. 

Let us give some more detail on this second part: In a first step, they consider spherically symmetric initial data on a hyperboloidal slice $\Sigma_0$ (which extends to $\mathcal{H}^+$ and terminates at $\mathcal{I}^+$) and assume that the following limit exists and is non-vanishing:
\begin{equation}\label{nt:eq:intro:r-2initialdata}
\lim_{r\to\infty}r^2\partial_r(r\phi_0)(u=0,r,\omega)=:I_0[\phi]<\infty.
\end{equation}
Now, the crucial observation is that the quantity 
\begin{equation}\label{nt:eq:intro:PNconstant}
\lim_{r\to\infty}r^2\partial_r(r\phi_0)(u,r,\omega)=:I_0[\phi](u)\equiv I_0[\phi]
\end{equation} (called the \textit{Newman--Penrose constant})
is, in fact, conserved along null infinity. It is this conservation law which is then exploited to derive the asymptotics of $r\phi_0$ in spacetime.

In a second step, they then consider spherically symmetric data for which $I_0[\phi]=0$, and require that, in the spirit of peeling (i.e.\ smoothness in the conformal variable $s=1/r$), 
\begin{equation}\label{nt:eq:intro:r-3initialdata}
\lim_{r\to\infty}r^3\partial_r(r\phi_0)(u=0,r,\omega)<\infty.
\end{equation}
There is no conservation law directly associated to this quantity. This difficulty is overcome by constructing the \textit{time integral} $\phi^{(1)}$ of $\phi$ (which satisfies $T\phi^{(1)}=\phi$, where $T$ is the stationary Killing field on Schwarzschild). It is shown that, generically, this time integral has a non-vanishing Newman--Penrose constant $I_0[\phi^{(1)}]$, which moreover can be computed from data for $\phi$. The authors of~\cite{Angelopoulos2018Late-timeSpacetimes} call this quantity the \textit{time-inverted Newman--Penrose constant}:
\begin{equation}\label{nt:eq:intro:PNconstanttimeinverted}
I_0[\phi^{(1)}]=:I_0^{(1)}[\phi].
\end{equation}
If this quantity is non-vanishing (which it is, generically), then one can apply the methods from the first step to $\phi^{(1)}$ in order to find its precise asymptotics, and then convert these asymptotics of $\phi^{(1)}$ to asymptotics of $\phi$ by commuting with $T$. This then proves Price's law.
If, on the other hand, $I_0^{(1)}[\phi]=0$, then one can construct the time integral of $\phi^{(1)}$ and proceed inductively to obtain faster decay.

Now, the conservation law \eqref{nt:eq:intro:PNconstant} is, in fact, a special case of the more general statement that, under suitable assumptions, 
\begin{equation}\label{nt:eq:intro:PNconstantf}
\lim_{r\to\infty}f(r)\partial_r(r\phi_0)(u,r,\omega)=:I^{f(r)}_0[\phi](u)
\end{equation}
is conserved along $\mathcal{I}^+$ if finite initially and if $f(r)/r^3\to 0$ as $r\to\infty$. 

We will be interested in the cases $f(r)=r^{-i}\log r$, $i=2,3$: 
Recall that the initial data we are interested in are to satisfy \eqref{nt:eq:intro:pvrphilog}.
 The modified Newman--Penrose constant associated to \eqref{nt:eq:intro:pvrphilog} is given by $\ilog[\phi]\equiv-2M\Phi^-$. 
 Even though this quantity is itself conserved along null infinity, it turns out to be easier to work with the associated modified Newman--Penrose constant of the time integral instead. We have the following relation:
\begin{equation}
\ilo[\phi^{(1)}]=-\ilog[\phi].
\end{equation}
In the main body of this paper, we will then present a modification of the argument in~\cite{Angelopoulos2018Late-timeSpacetimes} that replaces \eqref{nt:eq:intro:r-2initialdata} with the condition
\begin{equation}\label{nt:eq:intro:r-2logr initialdata}
0\neq\lim_{r\to\infty}\frac{r^2}{\log r}\partial_r(r\phi_0)(u=0,r,\omega):=\ilo[\phi]:=\ill[\phi]<\infty.
\end{equation}
This will allow us to show a \textit{logarithmically modified Price's law} for the $\ell=0$-mode, see Thm.~\ref{nt:thm:main}. 

\begin{rem}[Higher $\ell$-modes]
Recall from the above that it was shown in~\cite{Angelopoulos2018ASpacetimes} that, in the setting of  compactly supported Cauchy data, higher $\ell$-modes generally decay at least half a power faster towards $i^+$ than the $\ell\!=\!0$-mode. 
However, the setting we are interested in (motivated by our results in~\cite{Kerrburger} and the upcoming~\cite{Kerrburger3}) is such that, on $\Sigma_0$, the $\ell>0$-modes decay to leading order like $\pv(r\phi_\ell)\sim r^{-2}$, whereas the $\ell\!=\!0$-mode decays like $\pv(r\phi_0)\sim r^{-3}\log r$ -- more than half a power faster than the $\ell\! >\! 0$-modes -- so one might think that the $\ell\!=\!0$-mode does not determine the leading-order asymptotics in our setting.
However, recent work by Angelopoulos, Aretakis and Gajic ~\cite{Dejantobepublished}  indicates that, even in this setting, one can still expect higher $\ell$-modes to decay slightly faster. In particular, one can still expect the asymptotics of the $\ell\! >\! 0$-modes  to be subleading compared to the asymptotics of the $\ell\!=\!0$-mode obtained in this paper. This will be discussed in detail in~\cite{Kerrburger3}, see also the remarks below Theorem~\ref{nt:thm:main}.
For now, we restrict our presentation to the $\ell\!=\!0$-mode.
\end{rem}

\subsection[The main result (Theorem~\ref{nt:thm:main})]{The main result}
Let us now state a rough version of the main result of this paper (see \S \ref{sec:coordinates} for our choice of coordinates). The precise statement is written down in Theorems~\ref{thm:asyprecise} and~\ref{thm:connectoin}. 
\begin{thm}\label{nt:thm:main}
Let $\mathcal{C}_{\mathrm{in}}=\{v=v_0\}$ be an ingoing null hypersurface starting from $\mathcal{I}^-$ and extending to $\mathcal{H}^+$, and let $\epsilon_\phi>0$. Assume spherically symmetric initial data $\phi$ for \eqref{waveequation} on a Schwarzschild background which satisfy\footnote{ If $f$  and $g$ are functions depending only on one variable $x$, we say $f=\mathcal O_k(g)$ if there exist uniform constants $C_j>0$ such that $|\partial_x^j f|\leq C_j|\partial_x^j g|$ for $j=0,\dots,k$.}
\begin{equation}
\frac{\partial_u(r\phi)}{\pu r}(u,v_0)=\frac{I^{(\mathrm{past})}_0[\phi]}{r^2}+\mathcal{O}_4(r^{-2-\epsilon_\phi})
\end{equation}
for $u<0$ and $I^{(\mathrm{past})}_0[\phi]\neq 0$, and which also satisfy the no incoming radiation condition
\begin{equation}
\lim_{u\to -\infty} r\phi(u,v)=0
\end{equation}
for all $v\geq v_0$.
Assume further that the data smoothly extend to $\mathcal{H}^+$ (or that an appropriate energy norm of $\phi$ is finite) on $\{v=v_0\}$.
Then, for all $u,v>0$, the solution satisfies the following asymptotics near $i^+$:
\begingroup
\allowdisplaybreaks
\begin{align}
\left|\phi|_{\mathcal{H}^+}(v)+8MI^{(\mathrm{past})}_0[\phi]\frac{\log(1+v)}{(1+v)^3}\right|\leq C(v+1)^{-3},\\
\left|\phi|_{r=\con}(\tau)+8MI^{(\mathrm{past})}_0[\phi]\frac{\log(1+\tau)}{(1+\tau)^3}\right|\leq C(\tau+1)^{-3},\\
\left|	r\phi|_{\mathcal{I}^+}(u)+4MI^{(\mathrm{past})}_0[\phi]\frac{\log(u+1)}{(u+1)^2}		\right|\leq C(u+1)^{-2},
\end{align}
\endgroup
where $C>0$ is a constant completely determined by data. 
Moreover, we have for all $u<\infty$ that
\begin{equation}
\lim_{v\to\infty} \frac{r^3}{\log r}\frac{\pv(r\phi)}{\pv r}(u,v)=-2MI_0^{(\mathrm{past})}[\phi].
\end{equation}
\end{thm}

We believe that a few remarks are in order:
\begin{itemize}[leftmargin=*]
\item The appearance of logarithmic terms in \textit{higher-order asymptotics} is well-known (see, e.g.,~\cite{Angelopoulos2019LogarithmicInfinity, gomezwinicourschmidt,BASKIN2018160}). 
Similarly, modifications to Price's law have also been derived for spacetimes with different asymptotics than Schwarzschild (see, e.g.,~\cite{ching,morgan}). 
In contrast, the statement of Theorem~\ref{nt:thm:main} is that, \textit{under physically motivated assumptions} (rather than assuming compact support or conformal smoothness on a Cauchy hypersurface), there are logarithmic corrections to Price's law \textit{at leading order}.
\item The above theorem is obtained for the wave equation on a fixed Schwarzschild background. 
However, it easy to see that the proof generalises to other spherically symmetric spacetimes, most notably the subextremal Reissner--Nordstr\"om spacetimes. 
Moreover, the methods presented in this paper  can easily be applied to~\cite{Angelopoulos:2018uwb} to also obtain similar results for \textit{extremal} Reissner--Nordstr\"om spacetimes. 
In this case, however, the asymptotics would depend crucially on the extension of the data to $\mathcal{H}^+$, in view of the \textit{Aretakis constant} along $\mathcal{H}^+$. 
See also~\cite{AAGPRL}.
The generalisation to Kerr, on the other hand, will be the subject of future work (see also the recent~\cite{DejanKerr} and~\cite{Hintz}).
\item The above theorem is formulated for initial data on an ingoing null hypersurface $\mathcal{C}_{\mathrm{in}}$, however, by the results of~\cite{Kerrburger}, an entirely analogous statement holds for boundary data on a past-complete timelike hypersurface $\Gamma$ as considered in~\cite{Kerrburger} (see section 5.8 therein) which are suitably extended to $\mathcal{H}^+$.
\item The above theorem is obtained for spherically symmetric solutions $\phi$. 
However, as was mentioned before, the results of~\cite{Dejantobepublished} indicate that, even without symmetry assumptions, Theorem~\ref{nt:thm:main} gives the precise asymptotics since the higher $\ell$-modes can be expected to decay faster. 
We will discuss the precise early- \textit{and} late-time asymptotics of higher $\ell$-modes in detail in the upcoming~\cite{Kerrburger3}. In fact, we will find various different scenarios in~\cite{Kerrburger3}: In the case of polynomially decaying boundary data on a timelike hypersurface $\Gamma$, one can expect to recover a logarithmically modified Price's law ($r\phi_\ell|_{\mathcal I^+}\sim u^{-2-\ell}\log u$) for each $\ell$-mode. In the case of polynomially decaying data on an ingoing null hypersurface $\mathcal C_{\mathrm{in}}$, however, we will find that all higher $\ell$-modes decay like $r\phi_\ell|_{\mathcal I^+}\sim u^{-2}$ along null infinity, i.e.\ one logarithm faster than the $\ell\!=\!0$-mode. Finally, in the case of smooth compactly supported scattering data on $\mathcal I^-$ and $\mathcal H^-$, we will find that $r\phi_\ell|_{\mathcal I^+}\sim u^{-2}$ \textit{for all $\ell$}! While the usual belief that higher $\ell$-modes decay faster towards $i^+$ is hence violated on $\mathcal I^+$ in these settings, it still holds true away from $\mathcal I^+$, e.g.\ on hypersurfaces of constant $r$. See~\cite{Kerrburger3} for details.

\item The above theorem, in principle, gives a tool to \textit{directly measure} the non-smoothness of future null infinity. (See also~\cite{AAGPRL,Kroon_measure} in this context.)
\end{itemize}

\subsection{Structure of the paper}
This paper is structured as follows: In \S \ref{sec:geometry}, we shall introduce the geometry of the Schwarzschild spacetime and write down useful foliations of it. In \S \ref{sec:prelims}, we then import the necessary theory for the wave equation and, in particular, the almost-sharp decay results of~\cite{Angelopoulos2018ASpacetimes}. In \S \ref{sec:asymptotics}, we shall derive the precise late-time asymptotics for $\phi$ in the case $\ill[\phi]\neq 0$. We then derive the time inversion theory for the case $\ill[\phi]=0$ and $\ilog[\phi]\neq 0$ in \S \ref{sec:timeinversion}. Combining \S \ref{sec:asymptotics} and \S \ref{sec:timeinversion} then allows us to derive the precise late-time asymptotics for $\phi$ in the case $\ill[\phi]=0$ and $\ilog[\phi]\neq 0$ in \S \ref{sec:asymptotics2}.
We finally connect the results of \S \ref{sec:asymptotics2} to our results obtained in~\cite{Kerrburger} and, thus, prove Theorem~\ref{nt:thm:main} in \S \ref{sec:connectiontokerrburger}.
We conclude the discussion of the linear wave equation by discussing higher-order asymptotics in \S \ref{sec:hot}.

Note that all the results of this paper are obtained for the \textit{linear} wave equation on Schwarzschild, despite the results of~\cite{Kerrburger} having been obtained for the coupled Einstein-Scalar field system as well. We therefore give two brief comments on potential extensions of the results of the present paper to the coupled Einstein-Scalar field system in \S \ref{sec:nonlinearcomments}.

\section{The geometric setting}\label{sec:geometry}

\subsection{The Schwarzschild spacetime manifold}\label{sec:coordinates}
We  closely follow~\cite{Angelopoulos2018ASpacetimes}, with some minor adaptations:

The Schwarzschild family of spacetimes $(\mathcal{M}_M,g_M)$, $M>0$, is given by the family of manifolds with boundary
\begin{equation*}
\mathcal{M}_M=\mathbb{R}\times[2M, \infty)\times \mathbb{S}^2,
\end{equation*}
covered by the coordinate chart $(v,r,\theta, \varphi)$ with $v\in\mathbb{R}$, $r\in[2M,\infty)$, $\theta\in(0,\pi)$ and $\varphi\in(0,2\pi)$, where $(\theta,\varphi)$ denote the standard spherical coordinates on $\mathbb{S}^2$, and by the family of metrics
\begin{equation}
g_M=-D(r)\dd v^2+2\dd v\dd r+r^2(\dd\theta^2+\sin^2\theta \dd\varphi^2),
\end{equation}
where 
\begin{equation}\label{eq:geometry:D}
D(r)=1-\frac{2M}{r}.
\end{equation}
Note that the vector field $T=\pv$ is a Killing vector field.
We denote the boundary $\{r=2M\}=\partial\mathcal{M}_M=:\Hp$ as the future event horizon.

Next, we introduce the \textit{tortoise coordinate} $r^*$ as
\begin{equation}\label{eq:geometry:tortoise}
r^*(r):=R+\int_R^r D^{-1}(r')\dd r'
\end{equation}
for some $R>2M$ and define 
\begin{equation}
u:=v-2r^*.
\end{equation}
This gives rise to a covering  $(u,v,\theta,\varphi)$ of $\mathcal{M}_M\setminus\Hp$ with $u\in(\infty, \infty)$, $v\in(-\infty, \infty)$. The horizon is then "at $u=\infty$". 
The metric in these double-null coordinates reads
\begin{equation}
g_M=-D(r)\dd u\dd v+r^2(\dd\theta^2+\sin^2\theta \dd\varphi^2).
\end{equation}
We will drop the subscript $M$ from now on.

With respect to the $(u,v)$-chart, we define the null vector fields
\begin{align*}
\underline{L}:=\pu, &&L:=\pv.
\end{align*}
We then have that $$T=L+\underline{L} $$ and, in $(v,r)$-coordinates, 
\begin{align*}
\underline{L}=-\frac{D}{2}\partial_r,&&L=\frac{D}{2}\partial_r+\pv.
\end{align*}

\begin{rem}
We can generalise our results to spacetimes which, instead of \eqref{eq:geometry:D}, have
$
D(r)=1-\frac{2M}{r}+\mathcal{O}_k(r^{-1-\gamma})
$
for $\gamma>0$ and for sufficiently large values of $k$, subject to the condition that these spacetimes satisfy certain Morawetz (integrated local energy decay) estimates (see sections~2.4.1 and~2.4.2 in~\cite{Angelopoulos2018Late-timeSpacetimes}). Note that the sub-extremal Reissner--Nordstr\"om spacetime is such a spacetime, so the results of the present paper also apply to sub-extremal Reissner--Nordstr\"om spacetimes. 
\end{rem}
\subsection{The spacelike-null foliation}
Let $h:[2M, \infty)\to\mathbb{R}_{\geq0}$ be a non-negative, piecewise smooth function satisfying
\begin{equation}\label{eq:h}
0\leq\frac{2}{D(r)}-h(r)=\mathcal{O}(r^{-1-\eta})
\end{equation}
for some constant $\eta>0$.  Let further $v_0>0$, and define $v_{\Sigma_0}(r)$ as well as the spherically symmetric hypersurface $\Sigma_0$ via
\begin{align}
v_{\Sigma_0}(r):=v_{0}+\int_{2M}^r h(r')\dd r',&&  \Sigma_0:=\{(v,r,\theta, \varphi)\,|\,v=v_{\Sigma_0}(r)\}
.\end{align}
By construction, $\Sigma_0$ is a spacelike-null hypersurface which crosses the event horizon and terminates at future null infinity (condition \eqref{eq:h} ensures that $v$  (or $u$) is monotonically increasing (or decreasing) in $r$ along $\Sigma_0$).

For the sake of simpler notation, we will from now on restrict to examples of $\Sigma_0$ which moreover satisfy  the following condition:
There exist $2M<r_\mathcal{H}<4M<r_\mathcal{I}$ and $v_0, u_0>0$ such that
\begin{align*}
&\Sigma_0\cap\{r\leq r_\mathcal{H}\}=N^\mathcal{H}_0:=\{v=v_0\}\cap\{r\leq r_\mathcal{H}\},\\
&\Sigma_0\cap\{r\geq r_\mathcal{I}\}=N^\mathcal{I}_0:=\{u=u_0\}\cap\{r\geq r_\mathcal{I}\},
\end{align*}
and such that, moreover, the part $\Sigma_0\cap\{r_\mathcal{H}< r< r_\mathcal{I}\}$ is strictly spacelike. Furthermore, after potentially redefining $u$ and $r^*$ from eq.\ \eqref{eq:geometry:tortoise}, we can choose $u_0=0$ and $r_{\mathcal{I}}=R$.
\begin{figure}[htbp]
  \floatbox[{\capbeside\thisfloatsetup{capbesideposition={right,top},capbesidewidth=4.3cm}}]{figure}[\FBwidth]
{\caption{Depiction of the spacelike-null foliation of the Schwarzschild manifold $\mathcal M_M$ by the hypersurfaces $\Sigma_\tau$.}\label{fig:2}}
  {\includegraphics[width = 200pt]{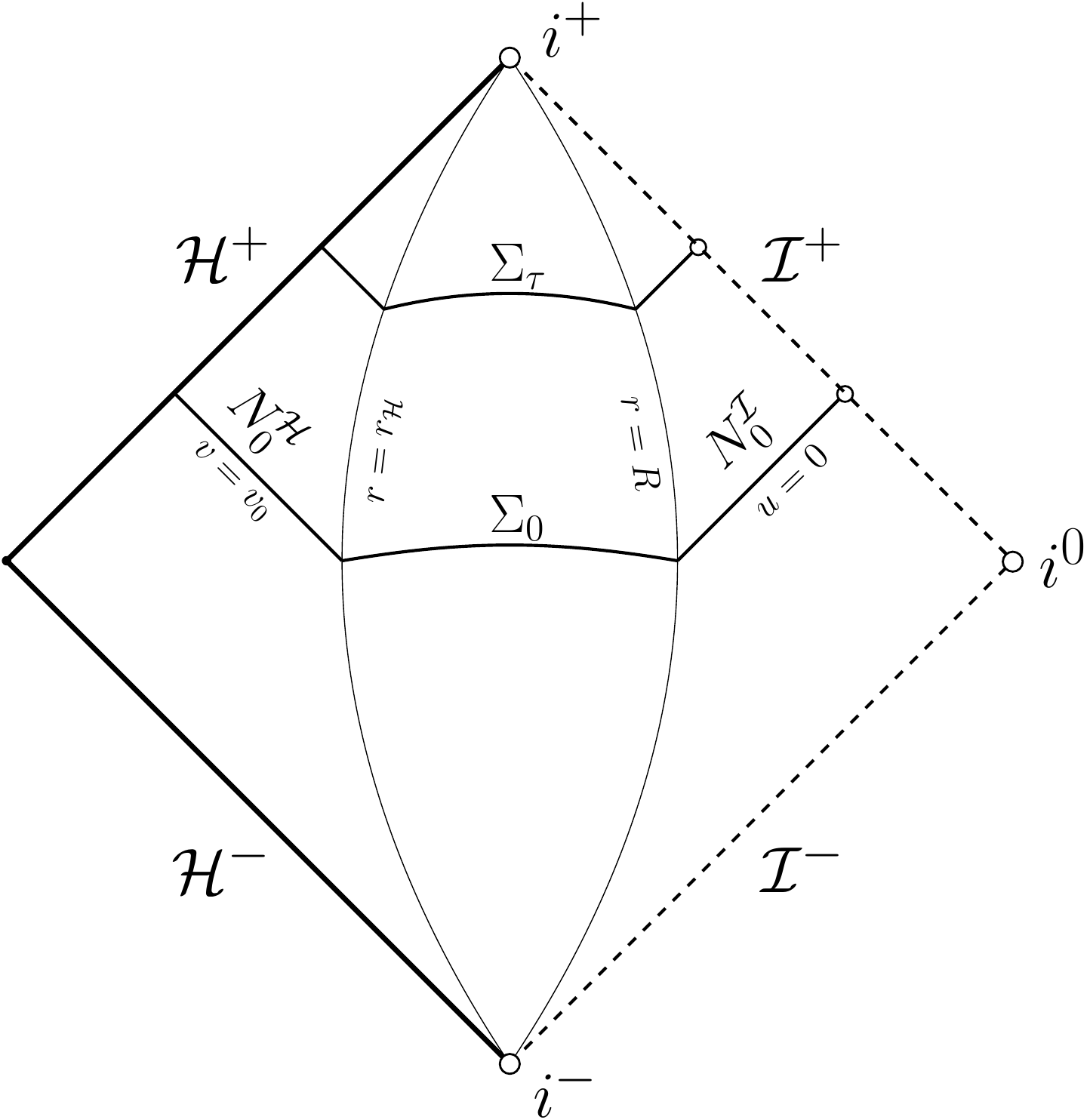}}
\end{figure}

Now, given $\Sigma_0$, we define a time function $\tau:J^+(\Sigma_0)\to\mathbb{R}_{\geq 0}$ via the flow of the stationary Killing field as follows:
\begin{align*}
\tau|_{\Sigma_0}=0,&&
T(\tau)=1.
\end{align*}
Let $F_\tau$ denote the flow of $T$, and define $\Sigma_\tau:=F_\tau(\Sigma_0)$. This gives rise to a spacelike-null foliation of $J^+(\Sigma_0)$ (see Figure~\ref{fig:2}). Adapted to this foliation, we can cover $J^+(\Sigma_0)$ with coordinates\footnote{Note that, for $\tau\geq 1$, we have $\tau\sim v $ for $r\leq r_{\mathcal{H}}$, $ \tau\sim v\sim u$ for $r_{\mathcal{H}}\leq r \leq r_\mathcal{I},$ and  $  \tau\sim u$ for $ r\geq r_\mathcal{I}$.
} $(\tau, \rho, \theta,\varphi)$ with $\rho|_{\Sigma_0}=r|_{\Sigma_0}$ and $\rho$ being constant along integral curves of $T$. In these coordinates, we have $T=\partial_\tau$, and the spherically symmetric vector field $Y$ tangent to $\Sigma_\tau$ is given by
\begin{align*}
Y=\partial_\rho=\partial_r+h\partial_v=-\frac{2}{D}\underline{L}+hT.
\end{align*}
We can then define the \textit{red-shift vector field} $N$ as follows:
\begin{align*}
N:=T-Y\ \ \text{in} \ \ \{2M\leq r\leq r_\mathcal{H}\},&&
N:=T \ \ \text{in} \ \ \{ r\geq r_\mathcal{I}\},
\end{align*}
with the additional requirement that the smooth matching in $r_\mathcal{H}\leq r\leq r_\mathcal{I}$ is such that $N$ remains time-invariant and strictly timelike.
\subsection{Notational conventions}
We use the notation $\dd\mu_{\Sigma_\tau}$ for the natural volume form on $\Sigma_\tau$ with respect to the induced metric, where, on the null parts of $\Sigma_\tau$, this volume form is chosen to be $r^2\dd\omega \dd u$, $r^2 \dd\omega \dd v$, respectively, with $\dd\omega=\sin\theta \dd\theta \dd\varphi$. Similarly, we denote the normal to $\Sigma_\tau$ by $n_{\Sigma_\tau}$, where we take the normals on the null parts to be $\underline{L}$, $L$, respectively.

We also say $f\sim g$ (or $f\lesssim g$) if there exists a uniform constant $C>0$ such that $C^{-1}g\leq f\leq Cg$ (or $f\leq C g$), and we use the usual algebra of constants ($C+D=C=CD\dots$).

\section{Preliminaries}\label{sec:prelims}
In this section, we recall the almost-sharp decay results obtained in~\cite{Angelopoulos2018ASpacetimes,Angelopoulos2018Late-timeSpacetimes}. We will first need to import some language.

\subsection{The Cauchy problem for the wave equation}
We recall the following standard result:
\begin{prop}\label{prop:Cauchy}
Let $\Phi\in C^\infty(\Sigma_0)$, $\Phi'\in C^\infty(\Sigma_0\cap\{r_\mathcal{H}< r< r_\mathcal{I}\})$. Then there exists a unique smooth function $\phi:J^+(\Sigma_0)\to\mathbb{R}$ satisfying
\begin{align*}
\phi|_{\Sigma_0}=\Phi,&&n_{\Sigma_0}(\phi)|_{\Sigma_0\cap\{r_\mathcal{H}< r< r_\mathcal{I}\}}=\Phi',
\end{align*}
and $$\Box_g \phi=0.$$
\end{prop}

\subsection{The modified Newman--Penrose constants \texorpdfstring{$I_0[\phi]$, $\ill[\phi]$ and $\ilog[\phi]$}{I0[phi], Ilog[phi] and Ilog2[phi]}}
Let $\phi$ be a solution to the wave equation in the sense of Proposition~\ref{prop:Cauchy}. Let moreover $f$ be a smooth function such that $\lim_{r\to\infty}r^{-3}f(r)=0$.
Then we define
\begin{equation}
I_0^f[\phi](u):=\frac{1}{4\pi}\lim_{r\to\infty}\int_{\mathbb{S}^2}f(r)\partial_r(r\phi)(u,r,\omega)\dd \omega.
\end{equation}
It is shown e.g.\ in\footnote{The proof there is only written for $f=r^{2}$, but it works for any smooth $f$ as specified above.}~\cite{Angelopoulos2018ASpacetimes} that this quantity, if finite initially, is, in fact, independent of $u$. In this case, we can write
\begin{equation}
I_0^{f}[\phi](u)=I^f_0[\phi](u_0)=:I^f_0[\phi].
\end{equation}
We moreover introduce the following notation:
\begin{align*}
I_0^{\frac{\log r}{r^2}}[\phi]:=\ill[\phi],&&
I_0^{\frac{1}{r^2}}[\phi]:=I_0[\phi].
\end{align*}
\paragraph{The past Newman--Penrose constant}
We finally define the past analogue of the Newman--Penrose constant $I_0$ for scalar fields $\phi$ which solve $\Box_g \phi=0$ on all of $\mathcal{M}$:
\begin{equation}
I_0^{(\mathrm{past})}[\phi](v):=\frac{1}{4\pi}\lim_{r\to\infty}\int_{\mathbb{S}^2}r^2\partial_r(r\phi)(v,r,\omega)\dd \omega.
\end{equation}

\subsection{The main energy norms}
In the sequel, we will refer to several initial data energy norms $E_k^N[\phi]$, $ E^\epsilon_{0,I_0^{\log}\neq 0; k}[\phi]$, $ \widetilde{E}^\epsilon_{0,I_0^{\log}\neq 0; k}[\phi]$, $ E^\epsilon_{0,I_0^{\log}= 0; k}[\phi]$ etc. 
These energy norms, which are defined on $\Sigma_0$, measure the almost-sharp decay (with an $\epsilon$-loss) and the regularity of the initial data on $\Sigma_0$, see already Propositions~\ref{prop5.2} and~\ref{cor7.6}.
Since they are only used for the black box results of \S \ref{sec:blackbox}, their definitions are deferred to appendix~\ref{app}. We remark already that, in the context of the scattering data we are ultimately interested in (which satisfy \eqref{nt:eq:intro:pvrphilog}), these energy norms will always be finite if enough regularity is assumed.

\subsection{The almost-sharp decay estimates}\label{sec:blackbox}
We have now introduced all the necessary baggage to finally quote the following two black box results (these correspond to Proposition~5.2 and Corollary~7.6 from~\cite{Angelopoulos2018Late-timeSpacetimes}, respectively):
\begin{prop}\label{prop5.2}
Let $\phi$ be a spherically symmetric solution of \eqref{waveequation} in the sense of Proposition~\ref{prop:Cauchy}, let $k\in\mathbb{N}_0$, and assume that $ E^\epsilon_{0,I_0^{\log}\neq 0; k+1}[\phi]<\infty$ for some $\epsilon\in(0,1)$.
Then there exists a constant $C(R,k,\epsilon)$ such that, for all $\tau\geq 0$:
\begin{align}
|T^k\phi|(\tau,\rho)&\leq C\sqrt{ E^\epsilon_{0,I_0^{\log}\neq 0; k+1}[\phi]}(1+\tau)^{-2-k+\epsilon},\label{eq:prop52phi}\\
\sqrt{\rho+1}\cdot|T^k\phi|(\tau,\rho)&\leq C\sqrt{ E^\epsilon_{0,I_0^{\log}\neq 0; k}[\phi]}(1+\tau)^{-\frac32-k+\epsilon}\label{eq:prop52sqrtrhophi},\\
\rho\cdot|T^k\phi|(\tau,\rho)&\leq C\sqrt{ E^\epsilon_{0,I_0^{\log}\neq 0; k}[\phi]}(1+\tau)^{-1-k+\epsilon}.\label{eq:prop52rhophi}
\end{align}
\end{prop}
\begin{prop}\label{cor7.6}
Let $\phi$ be a spherically symmetric solution of \eqref{waveequation} in the sense of Proposition~\ref{prop:Cauchy}, let $k\in\mathbb{N}_0$, and assume that $ \widetilde{E}^\epsilon_{0,I_0^{\log}\neq 0; k+1}[\phi]<\infty$ for some $\epsilon\in(0,1)$.
Then there exists a constant $C(R,k,\epsilon)$ such that, for all $\tau\geq 0$:
\begin{align}\label{eq:cor76}
\sqrt{\rho+1}(|NT^k (r\phi)|+|YT^k (r\phi)|)(\tau,\rho)\leq C\sqrt{\widetilde{E}^\epsilon_{0,I_0^{\log}\neq 0; k+1}[\phi]}(1+\tau)^{-\frac{5}{2}+\epsilon}.
\end{align}
\end{prop}

\section{Asymptotics I: The case  \texorpdfstring{$\ill[\phi]\neq0$}{Ilog[phi]!=0}}\label{sec:asymptotics}
In this section, we derive the precise late-time asymptotics for  spherically symmetric solutions $\phi$ to \eqref{waveequation}, evolving from initial data as in Proposition~\ref{prop:Cauchy}, which  have finite $\ill[\phi]\neq 0$. Let us from now on denote the causal future of $\Sigma_0$ as $\mathcal R$,  $J^+(\Sigma_0)=:\mathcal{R}$. 

We follow very closely section 8 of~\cite{Angelopoulos2018Late-timeSpacetimes}. 
Even though the methods are essentially identical, all of the proofs in~\cite{Angelopoulos2018Late-timeSpacetimes} require some adjustments in order to work in the case $\ill[\phi]\neq 0$ (remember that in~\cite{Angelopoulos2018Late-timeSpacetimes}, it is assumed that $I_0[\phi]<\infty$). 
Since those parts which do not require adjustments usually make up for just a few lines, we here opt for a mostly self-contained presentation rather than frequently referring to~\cite{Angelopoulos2018Late-timeSpacetimes}. 
Nevertheless, certain parts of our proofs will have a more detailed explanation in~\cite{Angelopoulos2018Late-timeSpacetimes}, in which case the reader will be informed of the precise reference.
\subsection{The splitting of the spacetime and the region \texorpdfstring{$\mathcal B_{\alpha}$}{B-alpha}}
We define, for $\alpha\in(0,1)$, the following subsets of $\mathcal{M}$:
\begin{align*}
\mathcal{B}_\alpha:=\{r\geq R\}\cap\{0\leq u\leq v-v^\alpha\}.
\end{align*}
We moreover denote 
\[\gamma_\alpha:=\{v-u=v^\alpha\}\cap\{u\geq 0\};\]
this is a timelike hypersurface which contains part of the boundary of $\mathcal{B}_\alpha$.
Without loss of generality, we assume that $v_{\gamma_\alpha}(u)\geq v_{r=R}(u)$ for all $u\geq 0$, where $v_{\gamma_\alpha}(u)$ is the unique $v$ such that $(u,v)\in \gamma_\alpha$, and we similarly define $v_{r=R}(u)$ and $u_{\gamma_\alpha}(v)$.

In the sequel, we will split up $\mathcal{R}$ into the regions $\mathcal{B}_\alpha$ for some suitable $\alpha$, $\mathcal{R}\cap\{r\geq R\}\setminus \mathcal{B}_\alpha$, and $\mathcal{R}\cap\{r\leq R\}$. See Figure~\ref{fig:3} below.
\begin{figure}[b]
  \floatbox[{\capbeside\thisfloatsetup{capbesideposition={right,top},capbesidewidth=4.8cm}}]{figure}[\FBwidth]
{\caption{Depiction of $\mathcal R:=J^+(\Sigma_0)$ and its subsets $\mathcal{B}_\alpha$ (to the right of the blue curve), $\mathcal{R}\cap\{r\geq R\}\setminus \mathcal{B}_\alpha$, and $\mathcal{R}\cap\{r\leq R\}$. The blue curve, in turn, corresponds to $\gamma_\alpha$. }\label{fig:3}}
  {\includegraphics[width = 200pt]{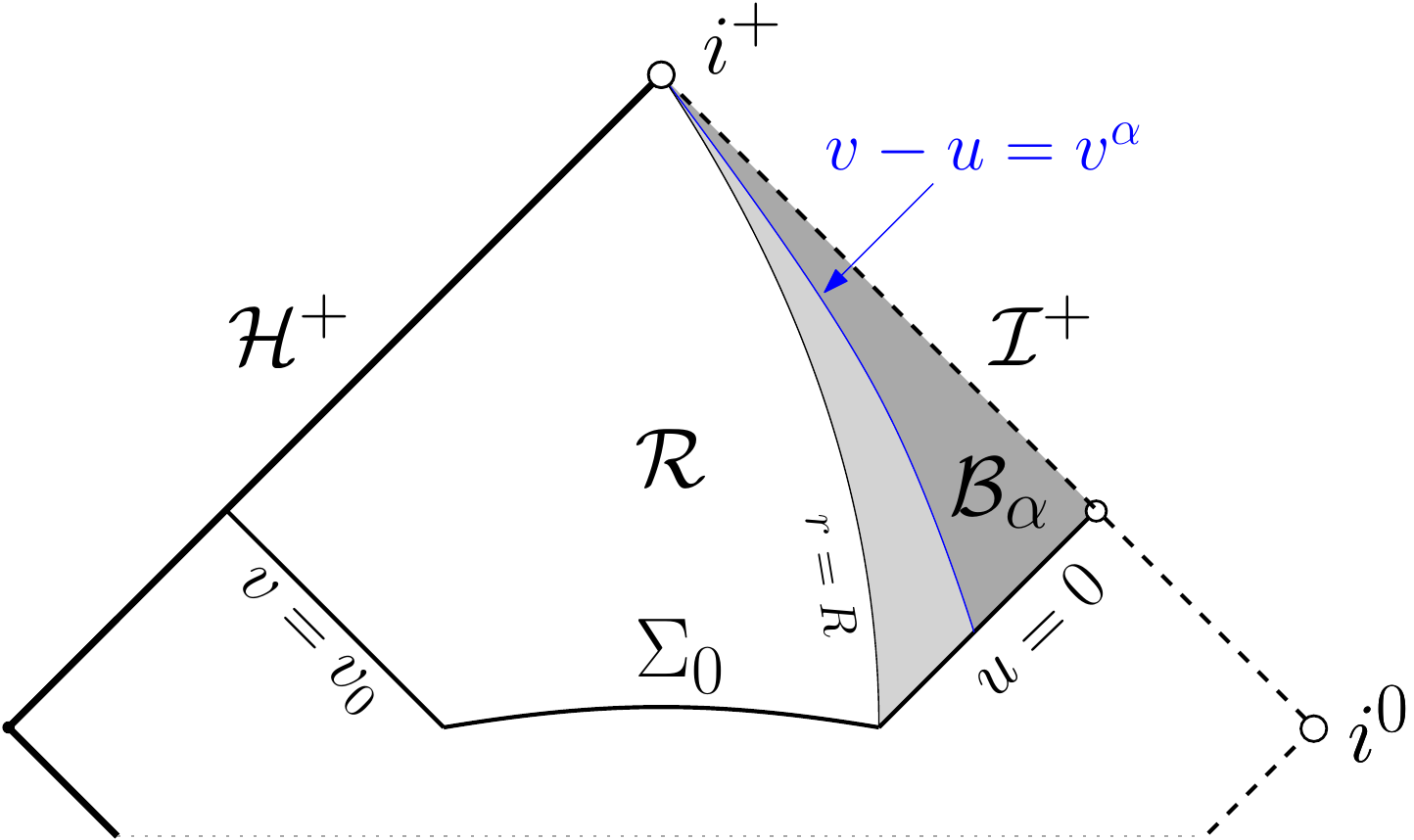}}
\end{figure}
For the reader's convenience, we here collect a few relations between $u,v$ and $r$ which will frequently be used in the following:
We have, throughout $\mathcal{B}_\alpha$, for sufficiently large $R$:
\begin{align}
r\gtrsim v-u\geq v^\alpha_{\gamma_\alpha}(u)\gtrsim(u+1)^{\alpha}\label{comp:rgreateru},\\
r\gtrsim v-u\geq v-u_{\gamma_\alpha}(v)=v^\alpha\label{comp:rgreaterv},\\
v\geq u+1\geq\frac{v_{\gamma_\alpha}(u)}{2}\geq \frac{u+1}{2}\label{comp:vgammalikeu}.
\end{align}
Moreover, we have throughout all of $\mathcal{R}\cap\{r\geq R\}$ that $\tau\sim u$ and that:
\begin{align}
|(v-u-1)-2r|\lesssim \log r\lesssim \log v\label{comp:rlikev-u-1},
\end{align}
and thus, in particular, $r\sim v-u$.
These relations can easily be checked using the definition of $\mathcal{B}_\alpha$ and eq.\ \eqref{eq:geometry:tortoise}. The implicit constants in $\sim$ and $\lesssim$ depend only on $M$ and $R$. Since $R>M$, they can, in fact, be chosen to depend only on $R$.

\subsection{Asymptotics for \texorpdfstring{$v^2\pv(r\phi)$}{v2 dv(rphi)} in the region \texorpdfstring{$\mathcal{B}_\alpha$}{B-alpha}}
\newcommand{\Plog}{P_{\ill,I_0',\beta}}
\newcommand{\Plogone}{P_{\ill,I_0',\beta;1}}
\newcommand{\Elog}{E^\epsilon_{0,\ill\neq0;0}}
Throughout the rest of this section, we assume that $\phi$ is a smooth, spherically symmetric solution arising from initial data on $\Sigma_0$. In addition to assuming that $\ill[\phi]<\infty$, it will be convenient to also assume that the following limit is finite on initial data:
\newcommand{\ifake}{I_0'[\phi]}
\begin{equation}\label{eq:ifake}
\lim_{r\to\infty}r^2\left(\partial_r(r\phi)(u_0,r)-\ill[\phi]\frac{\log r-\log 2}{r^2}\right):=\ifake.
\end{equation}
Let us then introduce the following $L^\infty$-norm on $\Sigma_0$ for $\beta>0$:
\begin{equation}\label{eq:definitionofP}
\Plog[\phi]:=\left|\left|v^{2+\beta}\left(\pv(r\phi)-2\frac{\ill[\phi]\log v}{v^2}-2\frac{\ifake}{v^2}\right)\right|\right|_{L^\infty(\Sigma_0)}.
\end{equation}
Our first proposition then concerns the asymptotics of $\pv(r\phi)$ in $\mathcal B_\alpha$:
\begin{prop}\label{prop8.1}
Let $\alpha\in(\frac{2}{3},1)$, $\epsilon\in(0,\frac{3\alpha-2}{2})$, and assume that $\Elog[\phi]<\infty$.
If there exists $\beta>0$ such that
\begin{equation}
\Plog[\phi]<\infty,
\end{equation}
then we have for all $(u,v)\in \mathcal{B}_\alpha$ that there exists a constant $C(R,\alpha,\epsilon)>0$ such that
\begin{equation}
|v^2\pv(r\phi)(u,v)-2\ill[\phi]\log v-2\ifake|\leq C\sqrt{\Elog[\phi]}\frac{1}{v^{3\alpha-2-2\epsilon}}+\Plog[\phi] v^{-\beta}.
\end{equation}
\end{prop}
\begin{proof}
The proof follows by integrating the wave equation for $r\phi$,
\begin{equation}\label{waveequationradiationfield}
\pu\pv(r\phi)=-\frac{DD'}{	4r}r\phi\left(\sim- \frac{r\phi}{r^3}\right)
\end{equation}
(which is implied by \eqref{waveequation} and where $'$ denotes $r$-differentiation), in $u$ from initial data.
This gives
\begin{align*}
|v^2\pv(r\phi)(u,v)-v^2\pv(r\phi)(0,v)|\leq C v^{-(3\alpha-2-2\epsilon)}&\int_0^u r^{-3}v^{3\alpha+2\epsilon}|r\phi|(u',v)\dd u'\\
\leq C\sqrt{\Elog[\phi]}v^{-(3\alpha-2-2\epsilon)}&\int_0^u(u'+1)^{-1-\epsilon}\dd u',
\end{align*}
 where we used the estimates \eqref{comp:rgreateru}, \eqref{comp:rgreaterv} and the almost-sharp decay estimate \eqref{eq:prop52rhophi} for $r\phi$ with $k=0$ (recall that $\tau\sim u$ in $\mathcal B_\alpha$). (Compare with the proof of Proposition~8.1 in~\cite{Angelopoulos2018Late-timeSpacetimes}.)
\end{proof}

\subsection{Asymptotics for the radiation field \texorpdfstring{$r\phi$}{rphi} in \texorpdfstring{$\mathcal{B}_\alpha$}{B-alpha}}
We will now use the asymptotics for $\pv(r\phi)$ obtained above to obtain decay for $r\phi$:
\begin{prop}\label{prop8.2}
Under the assumptions of Proposition~\ref{prop8.1}, with additionally $\alpha\in[\frac57,1)$ and $\epsilon\in(0,\frac16(1-\alpha))$, we have for all $(u,v)\in \mathcal{B}_\alpha$ that
\begin{align}\label{eq:prop8.2,1}
\begin{split}
&\left|r\phi(u,v)-2\ill[\phi]\left(\frac{\log(u+1)+1}{u+1}-\frac{\log (v)+1}{v}\right)-2\ifake\left(\frac{1}{u+1}-\frac{1}{v}\right)\right|\\
 \leq &C\left(\sqrt{\Elog[\phi]}+\ill[\phi]+\ifake\right)(u+1)^{\frac{\alpha}{2}-\frac32+2\epsilon}+C\Plog[\phi](u+1)^{-1-\beta},
\end{split}
\end{align}
where $C=C(R, \epsilon, \alpha)>0$ is a constant. 
 In fact, if we further impose $\frac{1-\alpha}{2}<\beta+2\epsilon$, then the estimate above provides asymptotics for $r\phi$ in the region $\mathcal{B}_\delta\subset\mathcal B_{\alpha}$, where $\delta$ is chosen such that $1>\delta>\frac{\alpha+1}{2}+2\epsilon>\alpha+2\epsilon$. 

In particular, setting $v=0$, the estimate  \eqref{eq:prop8.2,1} provides us with asymptotics for $r\phi$ along $\mathcal I^+$.
%
\end{prop}

\begin{proof}
Using the fundamental theorem of calculus, we write
\begin{equation}\label{eq:boundaryterm}
	r\phi(u,v)=r\phi(u,v_{\gamma_\alpha}(u))+\int_{v_{\gamma_\alpha}(u)}^v\pv(r\phi)(u,v')\dd v'.	
\end{equation}
The boundary term can be bounded by writing $r\phi=r^{\frac12} r^{\frac12}\phi$, writing $r^{\frac12}|_{\gamma_\alpha}\sim v_{\gamma_\alpha}^{\frac{\alpha}{2}}\sim(u+1)^{\frac{\alpha}{2}}$ by virtue of \eqref{comp:rlikev-u-1}, \eqref{comp:vgammalikeu} and the definition of $\gamma_\alpha$, and finally using the almost-sharp decay estimate \eqref{eq:prop52sqrtrhophi} with $k=0$. We thus obtain:
\[|r\phi(u,v_{\gamma_\alpha})|\leq C\sqrt{\Elog[\phi]}(u+1)^{\frac{\alpha}{2}-\frac32+\epsilon}.\]

In order to estimate the integral term, we plug in the result from the previous Proposition~\ref{prop8.1}, resulting in the estimate:
\begin{align}\label{proof:eq:prop8.2,1}
\begin{split}
\left|\int_{v_{\gamma_\alpha}}^v \pv(r\phi)(u,v')\dd v'-2\ill[\phi]\left(\frac{\log v_{\gamma_\alpha}}{v_{\gamma_\alpha}}-\frac{\log v}{v}\right)-2(\ill[\phi]+\ifake)\left(\frac{1}{v_{\gamma_\alpha}}-\frac{1}{v}\right)\right|\\
\leq C\sqrt{\Elog[\phi]}(v_{\gamma_\alpha}^{-3\alpha+1+2\epsilon}-v^{-3\alpha+1+2\epsilon})+C\Plog[\phi](v_{\gamma_\alpha}^{-1-\beta}-v^{-1-\beta}).
\end{split}
\end{align}
We first bound the terms from the LHS above. We write
\[v_{\gamma_\alpha}(u)^{-1}-v^{-1}=(v_{\gamma_\alpha}(u)^{-1}-(u+1)^{-1})+((u+1)^{-1}-v^{-1})\]
and estimate, using \eqref{comp:rgreateru},
\[		|v_{\gamma_\alpha}(u)^{-1}-(u+1)^{-1}|\leq(u+1)^{-1}v_{\gamma_\alpha}(u)^{\alpha-1}\leq C(u+1)^{-2+\alpha}.	\]
Similarly, we write 
\[\frac{\log v_{\gamma_\alpha}}{v_{\gamma_\alpha}}-\frac{\log v}{v}=\left(\frac{\log v_{\gamma_\alpha}(u)}{v_{\gamma_\alpha}(u)}-\frac{\log(u+1)}{u+1}\right)+\left(\frac{\log(u+1)}{u+1}-\frac{\log v}{v}\right)\]
and estimate
\begin{align*}
\frac{\log v_{\gamma_\alpha}(u)}{v_{\gamma_\alpha}(u)}-\frac{\log(u+1)}{u+1}=\log v_{\gamma_\alpha}(u)\left(\frac{1}{v_{\gamma_\alpha}}-\frac{1}{u+1}\right)+\frac{1}{u+1}\log \frac{v_{\gamma_\alpha}(u)}{u+1}\\
\leq C\frac{\log(u+1)}{(u+1)^{2-\alpha}}+\frac{1}{u+1}\log \left(1+\frac{v_{\gamma_\alpha}-u-1}{u+1}\right)\leq C\frac{\log(u+1)}{(u+1)^{2-\alpha}},
\end{align*}
where, in order to obtain the last inequality, we used
\[\log \left(1+\frac{v_{\gamma_\alpha}-u-1}{u+1}\right)\leq \frac{v_{\gamma_\alpha}-u-1}{u+1}\leq\frac{v_{\gamma_\alpha}^\alpha}{u+1}\leq C \frac{(u+1)^\alpha}{u+1}. \]

On the other hand, we estimate the terms on the RHS of \eqref{proof:eq:prop8.2,1} via
\[v_{\gamma_\alpha}(u)^{-1-\beta}-v^{-1-\beta}\leq v_{\gamma_\alpha}(u)^{-1-\beta}\leq (u+1)^{-1-\beta}, \]
and identically for the $(v_{\gamma_\alpha}^{-3\alpha+1+2\epsilon}-v^{-3\alpha+1+2\epsilon})$-term.

Finally, we can insert the estimates above back into \eqref{proof:eq:prop8.2,1} to find
\begin{align*}
\begin{split}
\left|\int_{v_{\gamma_\alpha}}^v \pv(r\phi)(u,v')\dd v'-2\ill[\phi]\left(\frac{\log(u+1)+1}{u+1}-\frac{\log (v)+1}{v}\right)-2\ifake\left(\frac{1}{u+1}-\frac{1}{v}\right)\right|\\
 \leq C\left(\sqrt{\Elog[\phi]}+\ill[\phi]+\ifake\right)(u+1)^{\frac{\alpha}{2}-\frac32+2\epsilon}+C\Plog[\phi](u+1)^{-1-\beta},
 \end{split}
\end{align*}
where we used that, for $\alpha\geq \frac57$, we have 
\[\max(-2+\alpha,1-3\alpha+2\epsilon)\leq \frac{\alpha}{2}-\frac32+2\epsilon.\]
This concludes the proof of the first statement \eqref{eq:prop8.2,1}.

To see that \eqref{eq:prop8.2,1} indeed gives the asymptotic behaviour in the region $\mathcal B_\delta$, we observe that, if $\epsilon\in(0,\frac{1}{6}(1-\alpha))$, we have $\frac{\alpha}{2}-\frac{3}{2}+2\epsilon<-1-\epsilon$. Furthermore, we have in the region $\mathcal{B}_\delta$ that
\begin{align*}
\left|\frac{1}{u+1}-\frac{1}{v}\right|\geq \left|\frac{1}{u+1}-\frac{1}{v_{\gamma_\delta}}\right|=\left|\frac{v_{\gamma_\delta}^\delta}{v_{\gamma_\delta}(u+1)}\right|\sim (1+u)^{-2+\delta}.
\end{align*}
Thus, if $1>\delta>\frac{\alpha+1}{2}+2\epsilon>\alpha+2\epsilon$, and if moreover $\frac{1-\alpha}{2}<\beta+2\epsilon$, we have $-2+\delta>-1-\beta$, and \eqref{eq:prop8.2,1} indeed gives the asymptotic behaviour in $\mathcal{B}_\delta$.
\end{proof}
\subsection{Asymptotics for \texorpdfstring{$T^k(r\phi)$}{Tk(rphi)} in the region \texorpdfstring{$\mathcal{B}_{\alpha_k}$}{B-alpha-k}}
We will now derive the asymptotics for $T^k(r\phi)$, $k>0$. This will be crucial later on when going back from the time integral of a solution to the original solution.

 In order to obtain the asymptotics for $T^k(r\phi)$ (Prop.~\ref{prop8.5}), we again first derive the asymptotics for $\pv(T^k(r\phi))$ (Prop.~\ref{prop8.4}). In turn, to derive the asymptotics for $\pv(T^k(r\phi))$, we will write $\pv T^k(r\phi)=\pv^{k+1}(r\phi)+\dots$, where the $\dots$-terms denote terms which, by the wave equation, decay faster. We will therefore first derive the asymptotics for $\pv^{k+1}(r\phi)$ in Proposition~\ref{prop8.3}.

In analogy to \eqref{eq:definitionofP}, we define  the following higher-order analogues of the norm $\Plog$ for $k\geq0$:
\newcommand{\Plogk}{P_{\ill,I_0',\beta;k}}
\newcommand{\Elogk}{E^\epsilon_{0,\ill\neq0;k}}
\begin{equation}\label{eq:definitionofPk}
\Plogk[\phi]:=\max_{0\leq j\leq k}\left|\left|v^{2+j+\beta}\pv^j\left(\pv(r\phi)-2\frac{\ill[\phi]\log v}{v^2}-2\frac{\ifake}{v^2}\right)\right|\right|_{L^\infty(\Sigma_0)}.
\end{equation}
We then have
\begin{prop}\label{prop8.3}
Let $k\in\mathbb{N}_0$, $\alpha_k\in(\frac{k+2}{k+3},1)$, let $\epsilon\in(0,\frac12(k+3)\alpha-\frac12(k+2))$, and assume that $\Elog[\phi]<\infty$. If moreover there exists $\beta>0$ such that
$
\Plogk[\phi]<\infty,
$
then we have for all $(u,v)\in \mathcal{B}_{\alpha_k}$:
\begin{align}
\begin{split}
&\left|	\pv^k \left(	\pv(r\phi)(u,v)-2\ill[\phi]\frac{\log v}{v^2}+2\ifake\frac{1}{v^2}\right)			\right| \\
			&\leq C \Plogk[\phi]v^{-2-k-\beta} +C\left(\sqrt{\Elog[\phi]} +\ill[\phi]+\ifake\right){v^{-(k+3)\alpha+2\epsilon}},
		\end{split}
\end{align}
where $C=C(R, \epsilon, \alpha_k,k)>0$ is a constant. 
\end{prop}
\begin{proof}
This proof proceeds in the same way as the proof of Proposition~\ref{prop8.1}, with the only difference being that we now inductively commute the wave equation \eqref{waveequationradiationfield} $k$ times with $\pv$, multiply it with $v^{k+2}$, and only then integrate in $u$. See the proof of Proposition~8.3 of~\cite{Angelopoulos2018Late-timeSpacetimes} for details.
\end{proof}

\begin{prop}\label{prop8.4}
Fix $k\in\mathbb{N}$. Under the assumptions of Proposition~\ref{prop8.3} and the additional assumption that $\Elogk[\phi]<\infty$, we have that
\begingroup\allowdisplaybreaks
\begin{align}\nonumber
&\left|		\pv T^k(r\phi)(u,v)-\pv^k\left(2\ill[\phi]\frac{\log v}{v^2}+2\ifake\frac{1}{v^2}\right)			\right| \\
			\leq& C \Plogk[\phi]v^{-2-k-\beta} +C\left(\sqrt{\Elog[\phi]} +\ill[\phi]+\ifake\right){v^{-(k+3)\alpha+2\epsilon}}\\
			 +&C\left(\sqrt{\Elogk[\phi]} +\ill[\phi]+\ifake\right)\sum_{l=0}^{k-1}r^{-3-l}(u+1)^{-k+l+\epsilon}\nonumber
\end{align}\endgroup
for all $(u,v)\in\mathcal{B}_{\alpha_k}$, where $C=C(R, \epsilon, \alpha_k,k)>0$ is a constant. 
\end{prop}

\begin{proof}
This proof is a consequence of the fact that
\[\pv T^k(r\phi)=\pv^{k+1}(r\phi)+\sum_{s=0}^{k-1}\sum_{\substack{l,m\geq 0;\\l+m=k-1-s}}\mathcal{O}(r^{-2-s})\pv^l T^m(r\phi),\]
combined with the results of the previous Proposition~\ref{prop8.3} and the estimate \eqref{eq:prop52sqrtrhophi} from Proposition~\ref{prop5.2}. See the proof of Proposition~8.4 in~\cite{Angelopoulos2018ASpacetimes} for details.
\end{proof}
Before we move on to the next proposition, we define a set of constants $c_k$ via the relations
\begin{equation}\label{eq:Ck}
\pv^{k-1}\left(\frac{\log v}{v^2}\right)=:(-1)^k k!\frac{c_k+\log v}{v^{k+1}}
\end{equation}
for $k\geq 1$ and set $c_0:=1$. Note that $c_1=0$.
\begin{prop}\label{prop8.5}
Fix $k\in\mathbb{N}$. Under the assumptions of Proposition~\ref{prop8.4} and the additional assumptions that $\alpha_k\in[\frac{2k+5}{2k+7},1)$ and $\epsilon\in(0,\frac{1}{6}(1-\alpha_k))$, we have that, for all $(u,v)\in\mathcal{B}_{\alpha_k}$,
\begin{align}\label{eq:prop85,1}
\begin{split}
	\left|	T^k(r\phi)(u,v)-2(-1)^k k!\left(\ill[\phi] \left(	\frac{\log(u+1)+c_k}{(u+1)^{k+1}}-\frac{\log(v)+c_k}{v^{k+1}}	\right)\right.\right.\ \ \ &\\
	\left.\left.+\ifake\left(	\frac{1}{(u+1)^{k+1}}-\frac{1}{v^{k+1}}\right)\right)	\right|&\\
	\leq C\left(\sqrt{\Elogk[\phi]} +\ill[\phi]+\ifake\right) (u+1)^{-\frac32-k+\frac{\alpha_k}{2}+2\epsilon}\\
	+C\Plogk[\phi](u+1)^{-1-k-\beta},
\end{split}
\end{align}
where $C=C(R, \epsilon, \alpha_k,k)>0$ is a constant and $c_k$ is defined in \eqref{eq:Ck}. 

In fact, if we further impose $\frac{1-\alpha_k}{2}<\beta+2\epsilon$, then the estimate above provides the asymptotics for $r\phi$ in the region $\mathcal{B}_{\delta_k}\subset \mathcal B_{\alpha_k}$ for $1>\delta_k>\frac{\alpha_k+1}{2}+2\epsilon>\alpha_k+2\epsilon$. 

In particular, we obtain the following asymptotics along $\mathcal{I}^+$:
\begin{align}
\begin{split}
	\left|	T^k(r\phi)(u,\infty)-2(-1)^k k!\left(\ill[\phi] 	\frac{\log(u+1)+c_k}{(u+1)^{k+1}}\right.\right.
	\left.\left.+\ifake	\frac{1}{(u+1)^{k+1}}\right)	\right|\\
	\leq C\left(\sqrt{\Elogk[\phi]} +\ill[\phi]+\ifake\right) (u+1)^{-1-k-\epsilon}\\
	+C\Plogk[\phi](u+1)^{-1-k-\beta}.
\end{split}
\end{align}
\end{prop}

\begin{proof}
The proof proceeds similarly to the proof of Proposition~\ref{prop8.2}. We already proved the case $k=0$ and can therefore restrict this proof to $k\geq 1$. We again apply the fundamental theorem of calculus in the $v$-direction, integrating from $\gamma_{\alpha_k}$:
\begin{equation}	T^k(r\phi)(u,v)=T^k(r\phi)(u,v_{\gamma_{\alpha_k}}(u))+\int_{v_{\gamma_{\alpha_k}}(u)}^v\pv T^k(r\phi)(u,v')\dd v.	\label{eq:funofcalc2}
\end{equation}

We use the result of Proposition~\ref{prop8.4} to estimate the integral term:
\begin{align*}
&\left|\int_{v_{\gamma_{\alpha_k}}(u)}^v\pv T^k(r\phi)(u,v')\dd v'		-2\left.\left(\pv^{k-1}\left(	\frac{\ill[\phi]\log v}{v^2}+\frac{\ifake}{v^2}	\right)\right)\right|^{v}_{v_{\gamma_{\alpha_k}}(u)}	\right|\\
\leq& C\left(\sqrt{\Elogk[\phi]}+\ill[\phi]+\ifake\right)(v_{\gamma_{\alpha_k}}^{1-(k+3)\alpha_k+2\epsilon}-v^{1-(k+3)\alpha_k+2\epsilon})\\
+&C\Plogk  (v_{\gamma_{\alpha_k}}^{-1-k-\beta}-v^{-1-k-\beta})\\
+&C\left(\sqrt{\Elogk[\phi]}+\ill[\phi]+\ifake\right)\sum_{l=0}^{k-1}\int_{v_{\gamma_{\alpha_k}}}^v r^{-3-l}(u+1)^{-k+l-\epsilon}\dd v'.
\end{align*}
We deal with the terms arising from the LHS by appealing to \eqref{eq:Ck} and using
\[|v_{\gamma_{\alpha_k}}(u)^{-k-1}-(u+1)^{-k-1}|\leq C(u+1)^{-2-k+\alpha_k}\]
as well as
\[|v_{\gamma_{\alpha_k}}(u)^{-k-1}\log( v)-(u+1)^{-k-1}\log(u+1)|\leq C(u+1)^{-2-k+\alpha_k}\log(u+1).\]
Estimating the first and the second term on the RHS is done as in the proof of Proposition~\ref{prop8.2}; we are hence left with the third term: Recalling \eqref{comp:rgreateru}, we find
\[		\sum_{l=0}^{k-1}\int_{v_{\gamma_{\alpha_k}}}^v r^{-3-l}(u+1)^{-k+l+\epsilon}\dd v'\leq  \sum_{l=0}^{k-1} \frac{C}{(u+1)^{k-l-\epsilon+(2+l)\alpha_k}}\leq C (u+1)^{-1+\epsilon-(k+1)\alpha_k}	.\]

On the other hand, to estimate the boundary term in \eqref{eq:funofcalc2}, we appeal to \eqref{eq:prop52sqrtrhophi}. This shows that the boundary term is bounded by (cf.\ \eqref{eq:boundaryterm})
\[|T^k(r\phi)(u,v_{\gamma_{\alpha_k}})|\leq C\sqrt{\Elogk[\psi]}(u+1)^{-\frac32-k+\frac{\alpha_k}{2}+\epsilon}.\]

The first statement of the proposition, \eqref{eq:prop85,1}, now follows since, in view of $\alpha_k\geq\frac{2k+5}{2k+7}$ and $\epsilon\in(0,\frac12(k+3)\alpha-\frac12(k+2))$, all the relevant exponents of $(u+1)$ appearing above are dominated by $-\frac{3}{2}-k+\frac{\alpha_k}{2}+2\epsilon$.

On the other hand, to see that \eqref{eq:prop85,1} provides asymptotics for $r\phi$ in the region $\mathcal B_{\delta_k}$ as specified, one proceeds as in the proof of Proposition~\ref{prop8.2}. Compare with the proof of Proposition~8.5 in~\cite{Angelopoulos2018Late-timeSpacetimes}.
\end{proof}

\subsection{Global asymptotics for the scalar field \texorpdfstring{$\phi$}{phi}}
In this section, we propagate the asymptotics obtained for $r\phi$ in $\mathcal{B}_\alpha$ into all of $\mathcal{R}$ and, in particular, into the region where $r\leq R$. 
In the region where $r$ is large, this requires another splitting into different spacetime regions. 
On the other hand, in the region where $r$ is small, we exploit that $\partial_\rho\phi$ exhibits good decay properties.
\newcommand{\Elogtilde}{\widetilde{E}^\epsilon_{0,\ill\neq0;0}}
\newcommand{\Elogtildeone}{\widetilde{E}^\epsilon_{0,\ill\neq0;1}}
\newcommand{\Elogtildek}{\widetilde{E}^\epsilon_{0,\ill\neq0;k}}
\newcommand{\Elogtildekplusone}{\widetilde{E}^\epsilon_{0,\ill\neq0;k+1}}
\begin{prop}\label{prop8.6}
Let $\epsilon\in(0,\min(\frac{1}{98},\beta))$, and assume that $\phi$ satisfies
\[\Plog[\phi]<\infty\]
as well as
\[\Elogtildeone[\phi]<\infty. \]

Then we have for all $(u,v)\in \mathcal{R}\cap\{r\geq R\}$:
\begin{align}
\begin{split}\label{eq:prop8.6,1}
	\left|	\phi(u,v)-\frac{4\ill[\phi]}{v-u-1}\left(	\frac{\log(u+1)}{u+1}-\frac{\log v}{v}	\right)	-4\frac{\ill[\phi]+\ifake}{(u+1)v}\right|&\\
	\leq C\left(	\sqrt{\Elogtildeone[\phi]}	+\ill[\phi]+\ifake+\Plog[\phi]	\right)(u+1)^{-1-\epsilon}v^{-1},
\end{split}
\end{align}
where $C=C(R,\epsilon)>0$ is a constant. On the other hand, we have, for another constant $C=C(R,\epsilon)>0$, in all of $\mathcal{R}\cap\{r\leq R\}$:
\begin{align}
\begin{split}
	\left| 	\phi(\tau,\rho)-4\ill[\phi]\frac{\log(\tau+1)}{(\tau+1)^2}-4\ifake\frac{1}{(\tau+1)^2}	\right|\\
	\leq C\left(	\sqrt{\Elogtildeone[\phi]}	+\ill[\phi]+\ifake+\Plog[\phi]	\right)(\tau+1)^{-2-{\epsilon}}.
\end{split}
\end{align}
\end{prop}

\begin{proof}
Let us first look at the case $(u,v)\in\mathcal{B}_\alpha$, with 
\begin{equation}\label{eq:alpha}
\frac{5}{7}<\alpha<1-6\epsilon
\end{equation}
for some $\epsilon\in(0,\min(\frac{1}{21},\beta))$ (which is four times the $\epsilon$ in the proposition).
 We essentially want to divide the estimate \eqref{eq:prop8.2,1} (which is not an asymptotic estimate in all of $\mathcal B_{\alpha}$!) by $r$. Recalling \eqref{comp:rlikev-u-1}, we have
\[\left|\frac{1}{2r}-\frac{1}{v-u-1}\right|\leq C\frac{\log v}{(v-u-1)^2}\]
and, using also \eqref{comp:vgammalikeu},
\[\left|\frac1r\left(\frac{1}{u+1}-\frac1v\right)-\frac{2}{(u+1)v}\right|\leq C\frac{\log v}{r(u+1)v} \leq \frac{C}{(v-u-1)(u+1)^{3/2}}.\]
Dividing now \eqref{eq:prop8.2,1} by $r$ and making use of the two estimates above and also the fact that $\beta>\epsilon$, we obtain that
\begingroup
\allowdisplaybreaks
\begin{align}\label{eq:Prop8.6proof,1}
\left|	\phi(u,v)-\frac{4\ill[\phi]}{v-u-1}\left(	\frac{\log(u+1)}{u+1}-\frac{\log v}{v}	\right)	-4\frac{\ill[\phi]+\ifake}{(u+1)v}\right|	\nonumber\\
\leq C\left(	\sqrt{\Elog[\phi]}+\ill[\phi]+\ifake+\Plog[\phi] 	\right)\frac{(u+1)^{\frac{\alpha-3}{2}+2\epsilon}}{v-u-1}\\
+C\ill[\phi]\frac{\log v}{(v-u-1)^2}\left(\frac{\log(u+1)}{u+1}-\frac{\log v}{v}\right).\nonumber
\end{align}\endgroup
In order to estimate the RHS, we need to restrict the region under consideration, namely $\mathcal{B}_\alpha$, to a smaller one, namely $\mathcal{B}_{\alpha+6\epsilon}$, and, moreover, partition this smaller region $\mathcal{B}_{\alpha+6\epsilon}$ into a region where $v$ is large and one where $v\sim u+1$:

\paragraph{Asymptotics in $\mathcal{B}_{\alpha+6\epsilon}\cap\{v-u-1>\frac{v}{2}\}$:}
In the region $\mathcal{B}_{\alpha+6\epsilon}\cap\{v-u-1>\frac{v}{2}\}$, we estimate the first term on the RHS of \eqref{eq:Prop8.6proof,1} as follows:
\[\frac{(u+1)^{\frac{\alpha-3}{2}+2\epsilon}}{v-u-1}\leq 2\frac{(u+1)^{\frac{\alpha-3}{2}+2\epsilon}}{v}<2v^{-1}(1+u)^{-1-\epsilon}.\]
Here, we used \eqref{eq:alpha} in the last estimate.
Similarly, we estimate the second term of the RHS in \eqref{eq:Prop8.6proof,1} via
\[\frac{\log v}{(v-u-1)^2}\left(\frac{\log(u+1)}{u+1}-\frac{\log v}{v}\right)\leq C \frac{\log v}{v^2}\frac{\log(u+1)}{u+1}\leq Cv^{-1}(1+u)^{-1-\epsilon},\]
where we converted some $v$-decay into $u$-decay in the last estimate.
This proves \eqref{eq:prop8.6,1} in the region $\mathcal{B}_{\alpha+6\epsilon}\cap\{v-u-1>\frac{v}{2}\}$.
\paragraph{Asymptotics in $\mathcal{B}_{\alpha+6\epsilon}\cap\{v-u-1\leq\frac{v}{2}\}$:}
In the region $\mathcal{B}_{\alpha+6\epsilon}\cap\{v-u-1\leq\frac{v}{2}\}$, we have, in particular, that $v\sim u+1$. Thus, using also that $v-u-1\gtrsim v^{\alpha+6\epsilon}$ by  definition of $\mathcal{B}_{\alpha+6\epsilon}$, we can estimate the first term on the RHS of \eqref{eq:Prop8.6proof,1} according to
\begin{align*}
\frac{(u+1)^{\frac{\alpha-3}{2}+2\epsilon}}{v-u-1}\leq C v^{-1}v^{1-\alpha-6\epsilon}(u+1)^{\frac{\alpha-3}{2}+2\epsilon}\leq C v^{-1}(u+1)^{-\frac{1+\alpha}{2}-4\epsilon}.
\end{align*}
If we, in addition to \eqref{eq:alpha}, also require that 
$
1-\alpha\leq 7\epsilon,
$
then we in fact have\footnote{Note that this calculation would not have worked in the region $\mathcal{B}_\alpha$, hence the restriction to $\mathcal{B}_{\alpha+6\epsilon}$.} 
\[\frac{(u+1)^{\frac{\alpha-3}{2}+2\epsilon}}{v-u-1}\leq Cv^{-1}(u+1)^{-1-\frac{\epsilon}{2}}.\]
As for the second term on the RHS of \eqref{eq:Prop8.6proof,1}, we simply write
\[\frac{\log v}{(v-u-1)^2}\left(\frac{\log(u+1)}{u+1}-\frac{\log v}{v}\right)	\leq 	C \frac{1}{v^{2\alpha+12\epsilon}}\frac{\log^2(u+1)}{u+1}\leq Cv^{-1}(1+u)^{-1-\epsilon}. \]
This proves \eqref{eq:prop8.6,1} in the region $\mathcal{B}_{\alpha+6\epsilon}\cap\{v-u-1\leq\frac{v}{2}\}$.
%
\paragraph{Asymptotics in $\mathcal{R}\cap\{r\geq R\}\setminus\mathcal{B}_{\alpha+6\epsilon}$:}
We will use the fundamental theorem of calculus, integrating \textit{inwards} along $N_\tau:=\{r\geq R\}\cap\{u=\tau\}$ from $\gamma_{\alpha+6\epsilon}$, recalling the almost-sharp decay estimate
\begin{align*}
r^{\frac12}|\pv \phi|\leq C \sqrt{\widetilde{E}^{\epsilon'}_{0,\ill\neq 0;1}}(1+u)^{-\frac{5}{2}+\epsilon'}
\end{align*}
for $\epsilon'=\epsilon/4$, which follows directly from \eqref{eq:cor76}. Fixing moreover now $$\alpha=1-7\epsilon,$$ we can then follow the exact same steps of the proof of Proposition~8.6, pp.~59-60 in~\cite{Angelopoulos2018Late-timeSpacetimes} to show that
\[|\phi(u,v)-\phi(u,v_{\gamma_{\alpha+6\epsilon}})|\leq C\sqrt{\widetilde{E}^{\epsilon'}_{0,\ill\neq 0;1}} v^{-1}(1+u)^{-1-\frac{\epsilon}{4}}. \]
Plugging in the asymptotics for $\phi(u,v_{\gamma_{\alpha+6\epsilon}})$, which we have obtained already, and also noting that $\Elog[\phi]\leq \widetilde{E}^{\epsilon'}_{0,\ill\neq 0;1}[\phi] $ for $\epsilon>\epsilon'$, we conclude the proof of \eqref{eq:prop8.6,1} (notice that the $\epsilon$ in the proposition corresponds to $\epsilon'$ in the proof).

\paragraph{Asymptotics in $\mathcal{R}\cap\{r\leq R\}$:}
We finally extend the asymptotics into the region where $r\leq R$. 
We first need to convert the $u$- and $v$-decay from \eqref{eq:prop8.6,1} into $\tau$-decay on $r=R$. By definition, we have on $r=R$ that $v-u=R$ and $\tau=u$. Therefore, we have on $r=R$:
\begin{align*}
\frac{1}{v-u-1}\left( \frac{\log(u+1)}{u+1}-\frac{\log v}{v}\right)&=	\log(u+1)\frac{1}{v(u+1)}	-\frac{\log(1+\frac{R-1}{u+1})}{v(v-u-1)}\\
&=\frac{\log(\tau+1)}{(\tau+1)^2}-\frac{1}{(\tau+1)^2} +\mathcal{O}\left(\frac{\log(1+\tau)}{(1+\tau)^3}\right),
\end{align*}
where we used a standard estimate for $\log(1+x)$ in the last line.
By the asymptotic estimate \eqref{eq:prop8.6,1}, we thus have
\begin{nalign}\label{ingredient1}
\left|\phi|_{\rho=R}(\tau)-4\ill[\phi]\frac{\log(\tau+1)}{(\tau+1)^2}-4\left((1-1)\ill[\phi]+\ifake\right)\frac{1}{(\tau+1)^2}	\right|\\
	\leq C\left(	\sqrt{\Elogtildeone[\phi]}	+\ill[\phi]+\ifake+\Plog[\phi]	\right)(\tau+1)^{-2-{\epsilon}}.
\end{nalign}
Finally, we have, by Proposition~\ref{cor7.6}, that 
\[\rho^{\frac12}|\partial_\rho \phi|\leq C\sqrt{\widetilde{E}^{\epsilon}_{0,\ill\neq0;1}}(1+\tau)^{-\frac52+\epsilon}.\]
Therefore, integrating along $\Sigma_\tau\cap\{r\leq R\}$ (recall that we set $r_{\mathcal{I}}=R$), we find
\begin{nalign}\label{ingredient2}
|\phi(\tau,\rho)-\phi(\tau,R)|&=\left|\int_\rho^R \partial_\rho \phi(\tau,\rho')\dd \rho'\right|\\
			\leq \int_{r}^R \rho^{-\frac12}\rho^{\frac12}|\partial_\rho\phi(\tau,\rho')\dd \rho'|
			&\leq C\sqrt{\widetilde{E}^{\epsilon}_{0,\ill\neq0;1}}(1+\tau)^{-\frac52+\epsilon}.
\end{nalign}
Combining \eqref{ingredient1} and \eqref{ingredient2} completes the proof of the proposition.
\end{proof}

\subsection{Global asymptotics for  \texorpdfstring{$T^k \phi$}{Tk phi}}
In order to apply our results to time derivatives of time integrals, we once again need to commute the global asymptotics of Proposition~\ref{prop8.6} with $T$.
\begin{prop}\label{prop8.7}
Let $k\in\mathbb{N}_0$. There exists an $\epsilon>0$ suitably small such that, under the assumptions
$\Elogtildekplusone[\phi]<\infty$
and
$\Plogk[\phi]<\infty$
for some $\beta>\epsilon$, we have that, for all $(u,v)\in \mathcal{R}\cap\{r\geq R\}$:
	\begin{multline}\left| T^k \phi(u,v)-4(-1)^k k!\left(	\frac{\ill[\phi]}{v-u-1}\left( \frac{\log(u+1)}{(u+1)^{k+1}}-\frac{\log v}{v^{k+1}}\right)	\right.\right.\ \ \ \\
	+\left.\left.\left( c_k\ill[\phi]+\ifake\right)\frac{1}{(u+1)^{k+1}v}\left(	1+\sum_{j=1}^k\left(\frac{u+1}{v}\right)^j	\right) 		\right)\right|\\
	\leq C\left(	\sqrt{\Elogtildekplusone[\phi]}	+\ill[\phi]+\ifake+\Plogk[\phi]	\right)(u+1)^{-k-1-\epsilon}v^{-1},
\end{multline}
where $C=C(R,k,\epsilon)>0$ and $c_k=c_k(k)$ are the constants defined in \eqref{eq:Ck}. In particular, $c_1=0$.

On the other hand, we have throughout $\mathcal{R}\cap\{r\leq R\}$ the estimate
\begin{align}
\begin{split}
	\left| 	T^k \phi(\tau,\rho)-4(-1)^k k!\left({\ill[\phi]}\frac{\log(\tau+1)}{(\tau+1)^{2+k}}+((c_k-1)\ill[\phi]+\ifake)\frac{k+1}{(\tau+1)^{2+k}}\right)	\right|\\
	\leq C\left(	\sqrt{\Elogtildekplusone[\phi]}	+\ill[\phi]+\ifake+\Plogk[\phi]	\right)(\tau+1)^{-2-k-\epsilon},
\end{split}
\end{align}
where $C=C(R,k,\epsilon)>0$ is a constant.
\end{prop}
\begin{proof}
The proof follows the same structure as the proof of Proposition~\ref{prop8.6}, now based on Proposition~\ref{prop8.4} instead of Proposition~\ref{prop8.2}. The only additional ingredient required is the identity
\begin{equation}
\frac{1}{(u+1)^{k+1}}-\frac{1}{v^{k+1}}=\frac{1}{(u+1)^{k+1}v^{k+1}(v-u-1)}\left(1+\sum_{j=0}^k (u+1)^jv^{k-j}\right).
\end{equation}
We refer the reader to the proof of Proposition~8.7 in~\cite{Angelopoulos2018Late-timeSpacetimes} for more  details.
\end{proof}

\section{Time inversion for \texorpdfstring{$\ill[\phi]=0$}{Ilog[phi]=0} and \texorpdfstring{$\ilog[\phi]\neq 0$}{Ilog2[phi]!=0}}\label{sec:timeinversion}
In the previous section, we have derived the precise late-time asymptotics for solutions with $\ill[\phi]\neq 0$. We now want to consider solutions with $\ill[\phi]=0$ and $\ilog[\phi]\neq 0$. As explained in the introduction, we can reduce to the case $\ill[\phi]\neq 0$ by considering the time integral $\phi^{(1)}$ of $\phi$. The purpose of this section is thus to extract the conditions needed on $\phi$ so that we can apply the results of the previous section to its time integral $\phi^{(1)}$.
\subsection{Construction of the time integral \texorpdfstring{$\phi^{(1)}$}{phi1}}
The approach of~\cite{Angelopoulos2018Late-timeSpacetimes} does not allow one to directly construct the time integral for data with $\ilog[\phi]\neq0$. Therefore, we follow the more elegant approach of~\cite{Angelopoulos:2018uwb}:
\begin{defi}
Let $\phi$ be a smooth, spherically symmetric solution to \eqref{waveequation} in the sense of Proposition~\ref{prop:Cauchy}, with $\ill[\phi]$ a well-defined limit.
We then define the \textit{time integral} $\phi^{(1)}$ of $\phi$ to be the unique spherically symmetric function $\phi^{(1)}:J^+(\Sigma_0)\to \mathbb{R}$ s.t.
\begin{align*}
T\phi^{(1)}=\phi, && \Box_g \phi^{(1)}=0,\\
\lim_{v\to\infty}\phi^{(1)}(u_0,v)=0,&& \lim_{u\to\infty}\underline{L}\phi^{(1)}(u,v_0)=0.
\end{align*}
\end{defi}
\begin{prop}
If $\phi$ is as in the definition above, then its time integral $\phi^{(1)}$ satisfies along $N_0^{\mathcal{I}}$ the following relation:
\begin{equation}
2r^2L\phi^{(1)}(u_0,v)=C'_0[\phi]+2\int_{N_0^\mathcal{I}\cap\{v'\leq v\}} rL(r\phi)(u_0,v')\dd v',
\end{equation}
where the constant $C'_0[\phi]$ is given by (writing $\Sigma_0\cap\{r_\mathcal{H}\leq r \leq r_\mathcal{I}\}=\Sigma_0^*$)
\begin{equation}
4\pi C'_0[\phi]:=2\int_{N_0^\mathcal{H}}r\underline{L}(r\phi)\dd u'+2\int_{\Sigma^*_0}n_{\Sigma_0}(\phi)\dd \mu_{\Sigma_0}+4\pi (r\phi|_{\Sigma_0\cap\{r=r_\mathcal{H}\}}+r\phi|_{\Sigma_0\cap\{r=r_\mathcal{I}\}}).
\end{equation}

Let us moreover assume that $\ilog[\phi]=\lim_{r\to\infty}\frac{r^3}{\log r}\partial_r(r\phi)<\infty$, and define 
$
4\pi C_0[\phi]:=4\pi C'_0[\phi]+2\int_{N_0^{\mathcal{I}}}rL(r\phi)(u,v')\, {\dd\omega}\dd v'
$.

Then we obtain the following additional relations along $N_0^\mathcal{I}$:
\begin{align}
	\begin{split}
\phi^{(1)}|_{N_0^\mathcal{I}}(r)&=-C_0\int_r^\infty \frac{1}{Dr'^2}\dd r'+2\int_r^\infty \frac{1}{Dr'^2}\int_{r'}^\infty r''\partial_r(r\phi)(u_0,r'')\dd r''\dd r',
	\end{split}
\end{align}\begin{align}
\label{eq:pr timeintegral}
	\begin{split}
	\partial_r(r\phi^{(1)})|_{N_0^\mathcal{I}}(r)&=C_0\left(\frac{1}{Dr}-\int_r^\infty \frac{1}{Dr'^2}\dd r'\right)\\
	+2\int_r^\infty& \frac{1}{Dr'^2}\int_{r'}^\infty r''\partial_r(r\phi)(u_0,r'')\dd r''\dd r'-\frac{2}{Dr}\int_{r}^\infty r'\partial_r(r\phi)(u_0,r')\dd r',
	\end{split}
\end{align}\begin{align}
\label{eq:pr2 timeintegral}
	\begin{split}
	\partial_r^2(r\phi^{(1)})|_{N_0^\mathcal{I}}(r)&=-C_0\frac{D'}{D^2r}
	+2\frac{D'}{D^2 r}\int_{r}^\infty r'\partial_r(r\phi)(u_0,r')\dd r'+\frac{2}{D}\partial_r (r\phi_0)|_{N_0^\mathcal{I}}(r).
	\end{split}
\end{align}

\end{prop}
\begin{proof}
A proof of the first statement is provided in Proposition~10.1 of~\cite{Angelopoulos:2018uwb}.  We rewrite it as
\[
L\phi^{(1)}(u_0,v)=\frac{C_0[\phi]}{2r^2}-\frac{1}{r^2}\int_v^\infty rL(r\phi)(u_0,v')\dd v'
.\]
We then switch to $(u,r)$-coordinates and integrate (recall that $L=\frac{D}{2}\partial_r$) the above equality from $\mathcal{I}^+$, where $\phi^{(1)}$ vanishes by definition, to obtain the second statement. 
The last two statements are then obtained by multiplying by $r$ and acting with $\partial_r$, $\partial_r^2$, respectively. (Recall that $D'=\frac{\dd}{\dd r} D(r)$ etc.)
\end{proof}
\subsection{The time-inverted Newman--Penrose constant \texorpdfstring{$I_0^{\log,(1)}[\phi]$}{ILog1[phi]}}
The following proposition expresses the Newman--Penrose constant $\ill[\phi^{(1)}]$ in terms of $\phi$:
\begin{prop}\label{prop:P}
Let $\phi$ be a smooth, spherically symmetric solution in the sense of  Proposition~\ref{prop:Cauchy}, and let $\beta\in(0,1)$ and $J$ be constants. Assume that, on initial data, $\phi$ satisfies for some constant $P$:
\begin{equation}\label{eq:prop:Pasymptoticsassumption}
\left|\partial_r(r\phi)|_{N_0^\mathcal{I}}(r)-\ilog[\phi] \frac{\log r}{r^3}-\frac{J}{r^3}\right|\leq P r^{-3-\beta}.
\end{equation}
Then there is a constant $C(R,P,\beta)$ such that the time integral $\phi^{(1)}$ of $\phi$ satisfies
\begingroup\allowdisplaybreaks
\begin{align}
\left|\partial_r(r\phi^{(1)})|_{N_0^\mathcal{I}}(r)+\ilog[\phi] \frac{\log r}{r^2}+\frac{J+\frac12 \ilog[\phi]-MC_0[\phi]}{r^2}\right|\leq Cr^{-2-\beta},\label{eq:P timeintegral}\\
\left|\partial_r^2(r\phi^{(1)})|_{N_0^\mathcal{I}}(r)-2\ilog[\phi] \frac{\log r}{r^3}-\frac{2J-2MC_0[\phi]}{r^3}\right|\leq Cr^{-3-\beta}\label{eq:P timeintegral pr2}.
\end{align}\endgroup
In particular, we have the following identities (recall the definition \eqref{eq:ifake}):
\begin{align}
I_0^{\log,(1)}[\phi]:=\ill[\phi^{(1)}]=-\ilog[\phi],\label{eq:I0logphi1}\\
I_0'[\phi^{(1)}]+\log (2)\ill[\phi^{(1)}]=MC_0[\phi]-J-\frac12 \ilog[\phi].\label{eq:I0primephi1}
\end{align}
\end{prop}

More generally, we can also show that if the asymptotics for $\partial_r(r\phi)$ above commute with $\partial_r^{k-1}$ on data, then the asymptotics for $\partial_r(r\phi^{(1)})$ commute with $\partial_r^k$ on data.

\begin{proof}
We obtain \eqref{eq:P timeintegral} by plugging the estimate \eqref{eq:prop:Pasymptoticsassumption} into identity \eqref{eq:pr timeintegral} and using that
\[\frac{1}{Dr}-\int_r^{\infty}\frac{1}{Dr'^2}\dd r'=\frac{M}{r^2}+\mathcal{O}(r^{-3}).\]
Similarly, we obtain \eqref{eq:P timeintegral pr2} by plugging  \eqref{eq:prop:Pasymptoticsassumption} into identity \eqref{eq:pr2 timeintegral}, noting that $D'=2Mr^{-2}$.
\end{proof}
We thus have as a direct corollary:
\begin{cor}\label{cor:Pfinite}
Under the assumptions of Proposition~\ref{prop:P}, we have that
\begin{equation}P_{\ill,I_0', \beta}[\phi^{(1)}]<\infty\end{equation}
 for $\ill[\phi^{(1)}]$ and $I_0'[\phi^{(1)}]$ as in \eqref{eq:I0logphi1}, \eqref{eq:I0primephi1}. Moreover, we have 
\begin{equation}P_{\ill,I_0' ,\beta,k}[\phi^{(1)}]<\infty\end{equation}
for $k=1$, where the norms $\Plog$, $\Plogk$ have been defined in eqns.\ \eqref{eq:definitionofP} and \eqref{eq:definitionofPk}, respectively.
\end{cor}

\subsection{Initial energy norms for \texorpdfstring{$\phi^{(1)}$}{phi1}}
Finally, in order to apply the results from section~\ref{sec:asymptotics} to time integrals $\phi^{(1)}$ of initial data with $\ilog[\phi]<\infty$, we need to estimate the relevant energy norms (namely $\Elogk$ and $\Elogtildek$) of $\phi^{(1)}$ in terms of initial data energy norms for $\phi$. As these energy norms are blind to logarithmic corrections, we can simply quote the following result from Proposition~9.6 in~\cite{Angelopoulos2018Late-timeSpacetimes}:
\begin{prop}\label{prop:energyfinite}
Let $k\in\mathbb{N}_0$, and let $\phi$ be a solution to the wave equation such that
\begin{equation}\widetilde{E}^{\epsilon}_{0,\ill=0;k}[\phi]+\int_{\Sigma_0}J^N[N^2\phi]\cdot n_{\Sigma_0} \dd \mu_{\Sigma_0}<\infty
\end{equation}
for some $\epsilon>0$. Then there exist a constant $C=C(R, \Sigma_0, \epsilon,k)>0$ such that the time integral $\phi^{(1)}$ of $\phi$ satisfies
\begin{align}
E^{\epsilon}_{0,\ill\neq 0;k+1}[\phi^{(1)}]\leq C\cdot \left(E^{\epsilon}_{0,\ill=0;k}[\phi]+\int_{\Sigma_0}J^N[N\phi]\cdot n_{\Sigma_0} \dd \mu_{\Sigma_0}\right),\\
\widetilde{E}^{\epsilon}_{0,\ill\neq 0;k+1}[\phi^{(1)}]\leq C\cdot \left(\widetilde{E}^{\epsilon}_{0,\ill=0;k}[\phi]+\int_{\Sigma_0}J^N[N^2\phi]\cdot n_{\Sigma_0} \dd \mu_{\Sigma_0}\right).
\end{align}
\end{prop}

\section{Asymptotics II: The case \texorpdfstring{$\ill[\phi]=0$}{Ilog[phi]=0} and \texorpdfstring{$\ilog[\phi]\neq 0$}{Ilog2[phi]!=0}}\label{sec:asymptotics2}
We can now combine the results of sections~\ref{sec:asymptotics} and~\ref{sec:timeinversion} to derive the precise late-time asymptotics for solutions coming from smooth spherically symmetric initial data with $\ilog[\phi]<\infty$. This is done by simply writing $\phi$ as a time derivative of its time integral $\phi^{(1)}$,
\[r\phi=T(r\phi^{(1)}),\]
for which then Propositions~\ref{prop8.4},~\ref{prop8.5} and~\ref{prop8.7} hold, assuming that the relevant energy and $L^\infty$-norms ($\Plog$ etc.) of $\phi^{(1)}$ are finite. This latter assumption can in turn be shown to hold using Proposition~\ref{prop:energyfinite} and Corollary~\ref{cor:Pfinite}, respectively. 

We summarise our findings in 
\begin{thm}\label{thm:asyprecise}
Let $\epsilon>0$ be sufficiently small. Let $\phi$ be a spherically symmetric solution to the wave equation arising from smooth initial data  on $\Sigma_0$ (in the sense of Proposition~\ref{prop:Cauchy}) with $\ilog[\phi]<\infty$, and assume that
\begin{equation}
\widetilde{E}^{\epsilon}_{0,\ill=0;1}[\phi]+\int_{\Sigma_0}J^N[N^2\phi]\cdot n_{\Sigma_0} \dd \mu_{\Sigma_0}<\infty.
\end{equation}
Assume moreover that there exist constants $J$, $P$ and $\beta\in(\epsilon,1)$ such that on initial data, for all $r\geq R$:
\begin{equation}
\left|\partial_r(r\phi)|_{N_0^\mathcal{I}}(r)-\ilog[\phi] \frac{\log r}{r^3}-\frac{J}{r^3}\right|\leq P r^{-3-\beta}.
\end{equation}
Then the time integral $\phi^{(1)}$ satisfies 
\begin{align}
\ill[\phi^{(1)}]&=-\ilog[\phi]\label{Thm:eq:illfinite},\\
I_0'[\phi^{(1)}]&=MC_0[\phi]-J-\left(\frac12+\log 2\right) \ilog[\phi]\label{Thm:eq:ifakefinite},
\end{align}
and
\begin{equation}\label{Thm:eq:PandEarefinite}
\Plog[\phi^{(1)}]+\Plogone[\phi^{(1)}]+E^\epsilon_{0,\ill\neq0;2}[\phi]+\widetilde{E}^\epsilon_{0,\ill\neq0;2}[\phi^{(1)}]<\infty,
\end{equation}
where the $P$-norms have been defined in \eqref{eq:definitionofP} and \eqref{eq:definitionofPk}, respectively.

Moreover, we have the following asymptotic estimates for $\phi$:
For all $(u,v)\in\mathcal{R}\cap\{r\leq R\}$, we have:
\begin{align}
	&\left| 	 \phi(\tau,\rho)+4\left({\ill[\phi^{(1)}]}\frac{\log(\tau+1)}{(\tau+1)^{3}}+\left(I_0'[\phi^{(1)}]-\ill[\phi^{(1)}]\right)\frac{2}{(\tau+1)^{3}}\right)	\right|\\
	\leq& C\left(	\sqrt{\widetilde{E}^\epsilon_{0,\ill\neq0;2}[\phi^{(1)}]}	+\ill[\phi^{(1)}]+I_0'[\phi^{(1)}]+\Plogone[\phi^{(1)}]	\right)(\tau+1)^{-3-\epsilon}.\nonumber
\end{align}
On the other hand, we have that, for all $(u,v)\in \mathcal{R}\cap\{r\geq R\}$:
\begin{align}
	&\left|  \phi(u,v)+	\frac{4\ill[\phi^{(1)}]}{v-u-1}\left( \frac{\log(u+1)}{(u+1)^{2}}-\frac{\log v}{v^{2}}\right)	+\frac{4I_0'[\phi^{(1)}]}{(u+1)^{2}v}\left(	1+\frac{u+1}{v}	\right) 		\right|\\
	\leq& C\left(	\sqrt{\widetilde{E}^\epsilon_{0,\ill\neq0;2}[\phi^{(1)}]}	+\ill[\phi^{(1)}]+I_0'[\phi^{(1)}]+\Plogone[\phi^{(1)}]	\right)(u+1)^{-2-\epsilon}v^{-1}.\nonumber
\end{align}
In particular, the following asymptotics hold along $\mathcal{I}^+$:
\begin{align}
	&\left|	r\phi(u,\infty)+2\left(\ill[\phi^{(1)}] 	\frac{\log(u+1)}{(u+1)^{2}}+I_0'[\phi^{(1)}]	\frac{1}{(u+1)^{2}}\right)	\right|\\
	\leq& C\left(\sqrt{E^\epsilon_{0,\ill\neq0;1}} +\ill[\phi^{(1)}]+I_0'[\phi^{(1)}]+\Plogone[\phi^{(1)}]\right) (u+1)^{-2-\epsilon}.\nonumber
\end{align}
In each case, $C=C(R,\epsilon)>0$ is a constant.
\end{thm}
\begin{proof}
By Proposition~\ref{prop:energyfinite} and Corollary~\ref{cor:Pfinite}, the statements \eqref{Thm:eq:illfinite}, \eqref{Thm:eq:ifakefinite} and \eqref{Thm:eq:PandEarefinite} for the time integral $\phi^{(1)}$ of $\phi$ follow. 
The remaining statements then follow by applying Propositions~\ref{prop8.5} and~\ref{prop8.7} to $\phi^{(1)}$ and $k=1$, recalling that $c_k=0$ for $k=1$.
\end{proof}

\section{Proof of Theorem~\ref{nt:thm:main}}\label{sec:connectiontokerrburger}
We finally apply our results to the data of~\cite{Kerrburger} and, thus, prove the main theorem (Thm.~\ref{nt:thm:main}):
\begin{thm}\label{thm:connectoin}
Consider smooth, spherically initial/scattering data for \eqref{waveequation} as in~\cite{Kerrburger}, that is, assume that, on $\mathcal{C}_{\mathrm{in}}:=\{v=v_0\}$,
\begin{equation}\label{xx}
\partial_u (r\phi)(u,v_0)=\frac{2I^{(\mathrm{past})}_0[\phi]}{(1+u)^2}+\mathcal{O}_4((1+u)^{-2-\epsilon_\phi})
\end{equation}
for $u<0$ and some $\epsilon_\phi>0$, and that
\begin{equation}
\lim_{u\to -\infty} r\phi(u,v)=0
\end{equation}
for all $v\geq v_0$.
Moreover, assume that the data on $\{v=v_0\}$ extend smoothly to the event horizon $\mathcal{H}^+$.

Then there exist constants $J$, $P>0$ and $\beta\in(0,\epsilon_\phi)$ such that the arising scattering solution satisfies
\begin{equation}\label{x}
\left|\partial_r(r\phi)|_{N_0^\mathcal{I}}(r)-\ilog[\phi] \frac{\log r}{r^3}-\frac{J}{r^3}\right|\leq P r^{-3-\beta},
\end{equation}
where $\ilog[\phi]$ is given by
\begin{equation}
\ilog[\phi]=-2M I_0^{(\mathrm{past})}[\phi].
\end{equation}
Moreover, we have on $\Sigma_0$ that
\begin{equation}\label{xt}
\widetilde{E}^{\epsilon}_{0,\ill=0;1}[\phi]+\int_{\Sigma_0}J^N[N^2\phi]\cdot n_{\Sigma_0} \dd \mu_{\Sigma_0}<\infty.
\end{equation}
In particular, Theorem~\ref{thm:asyprecise} applies to $\phi$.
\end{thm}
\begin{rem}
An almost identical statement holds for the boundary value problem with polynomially decaying data on a spherically symmetric past-complete timelike hypersurface $\Gamma$ as considered in~\cite{Kerrburger}, see Theorem~2.1 therein. Moreover, a more detailed analysis shows that $\beta$ can be chosen to equal $\epsilon_\phi$ if $\epsilon_\phi<1$, and that the RHS of \eqref{x} can be replaced by $Pr^{-4}\log r$ if $\epsilon_\phi\geq1$. Lastly, it suffices to assume finite regularity for the data on $v=v_0$ and to assume that the energy expression 
\[\sum_{i\leq 8,j\leq1} \int_{\{v=v_0\}} \left(J^N[T^i\phi]+J^N[NT^j\phi]+J^N[N^2\phi]\right)\cdot \underline{L}\dd u\] remains locally finite in a neighbourhood of $\mathcal H^+$.
\end{rem}
\begin{proof}
We will, in this proof, change coordinates $(\bar{u},\bar{v})=(\frac{u}{2},\frac{v}{2})$. Moreover, to match with the notation of~\cite{Kerrburger}, we will also write $\Phi^-:=I_0^{(\mathrm{past})}[\phi]$.
By the results of~\cite{Kerrburger} (Theorem~2.4 or, more precisely, Theorems~4.1 and~4.2 therein), a unique smooth solution assuming the initial/scattering data exists up to $\Sigma_0$, and we have, for sufficiently large negative values of $u$, that
\begin{equation}
\left|r\phi(u,v)-\frac{\Phi^-}{|\bar{u}|}\right|\leq C|\bar u|^{-1-\epsilon_\phi}
\end{equation}
for some uniform constant $C$.

In order to prove \eqref{x}, we need to re-examine the steps of the proof of Theorem~4.2. By inserting the above estimate into the wave equation~\eqref{waveequationradiationfield} and integrating the latter from past null infinity, we obtain that
\begin{align*}
\left|	\partial_{\bar{v}}(r\phi)(\bar{u},\bar{v})+2M\int_{-\infty}^{\bar{u}} D\frac{\Phi^-}{r^3|\bar{u}'|}\dd \bar{u}'\right|\leq C r^{-3}|\bar{u}|^{-\epsilon_\phi}
\end{align*}
for sufficiently large negative values of $u$.
Consider now the expression
\begin{nalign}
&\partial_{\bar{u}}\left(r^3\partial_{\bar{v}}(r\phi)+2Mr^3\int_{-\infty}^{\bar{u}} D\frac{\Phi^-}{r^3|\bar{u}'|}\dd \bar{u}'\right)\\
=&\underbrace{2MD\left(	\frac{\Phi^- }{|\bar{u}|}-r\phi	\right)}_{\leq C|u|^{-1-\epsilon_\phi}}-\underbrace{3Dr^2\left(	\partial_{\bar{v}}(r\phi)+2M\int_{-\infty}^{\bar{u}} D\frac{\Phi^-}{r^3|\bar{u}'|}\dd \bar{u}'\right)}_{\leq Cr^{-1}|u|^{-\epsilon_\phi}},
\end{nalign}
and, after estimating, say, $r^{-1+\beta}$ against $|u|^{-1+\beta}$ for some $\beta<\epsilon_\phi$, integrate it from $-\infty$ to some fixed value $\bar{u}$ to obtain the more precise asymptotic expression:
\begin{equation}\label{asy2}
\left|r^3	\partial_{\bar{v}}(r\phi)(\bar{u},\bar{v})+2Mr^3\int_{-\infty}^{\bar{u}} D\frac{\Phi^-}{r^3|\bar{u}'|}\dd \bar{u}'	-2M\int_{-\infty}^{\bar{u}}\left(\frac{\Phi^-}{|\bar{u}'|}-r\phi(\bar{u}', \bar{v})		\right) \dd \bar{u}'		\right|\leq \frac{C}{r^\beta}.
\end{equation}
As in the proof of Theorem 4.2 in~\cite{Kerrburger}, we evaluate the integral on the LHS by using $|\bar{v}-\bar{u}-r|\lesssim \log r$ and decomposing into partial fractions.
 We then have, \textit{for fixed $\bar u$},
\[\int_{-\infty}^{\bar u}\frac{D}{r^3|\bar{u}'|}\dd \bar{u}'=\int_{-\infty}^{\bar{u}} \frac{1}{(\bar{v}-\bar{u}')^3|\bar{u}'|}\dd  \bar{u}'+\mathcal{O}\left(\frac{1}{r^{3+\beta}}\right)=\frac{\log r}{r^3}+\frac{J(\bar{u})}{r^3}+\mathcal{O}\left(\frac{1}{r^{3+\beta}}\right)\]
for some function $J(\bar u)$.

Inserting this estimate back into \eqref{asy2} gives the asymptotics for $\pv(r\phi)$ for fixed, sufficiently large negative values of $\bar{u}$. 
However, in view of the wave equation for the radiation field \eqref{waveequationradiationfield}, we can, in fact, propagate them for any finite $u$-distance; in particular, we can propagate them up to $\Sigma_0$. This proves \eqref{x}.

It is left to prove \eqref{xt}. Proving the finiteness of the $J^N$-based energies contained in the definition of $\widetilde{E}^{\epsilon}_{0,\ill=0;1}[\phi]$ is standard.
On the other hand, to show the finiteness of terms like e.g.\
\[\int_{N_0^\mathcal{I}} r^{5+2k-\epsilon}(\partial_r^{1+k}(r\phi))^2\dd r\]
for $k=1$, we use the asymptotic expression \eqref{x} as well as the fact that the 
asymptotics obtained in Theorem~4.2 of~\cite{Kerrburger} commute with $\pv$: For instance, we have that $\pv^2(r\phi)\sim r^{-4}\log r$ (see Remark~4.5 in~\cite{Kerrburger}).
Lastly, we similarly deal with terms such as
\[\int_{N_0^\mathcal{I}} r^{4-\epsilon}(\partial_r^{1}T^3(r\phi))^2\dd r\]
by writing $T=\pu+\pv$ and using the wave equation \eqref{waveequationradiationfield} for $\pu\pv$-terms and the improved estimates mentioned above for $\pv\pv$-terms etc. 
\end{proof}

\section{Comments on higher-order asymptotics}\label{sec:hot}
In this section, we shall briefly discuss the issue of deriving higher-order asymptotics for $\phi$. 

In the settings where the solution is conformally smooth on $u=0$, i.e.,
	\begin{equation}\label{eq:hot:1}
		\partial_r(r\phi)(0,r)=\frac{I_0}{r^2}+\mathcal{O}(r^{-3}) 
	\end{equation}
or
	\begin{equation}\label{eq:hot:2}
		\partial_r(r\phi)(0,r)=\frac{J_0}{r^3}+\mathcal{O}(r^{-4}) 
	\end{equation}
	for some constants $I_0$, $J_0$, the next-to-leading-order asymptotics for $\phi$ have been derived in~\cite{Angelopoulos2019LogarithmicInfinity}. It was  found there that these next-to-leading-order asymptotics contain logarithmic terms which, on the one hand, have contributions from the leading-order asymptotics of the solution itself. On the other hand, they have contributions from the simple fact that, on Schwarzschild, assuming smoothness in the variable $s=1/r$ is incompatible with assuming smoothness in the variable $s'=1/v$. (This is because $r+2M\log r- (v-u)/2=\mathcal{O}(1)$.) In other words, the above initial data assumptions imply
	\begin{equation}
		\pv(r\phi)(0,r)=\frac{2I_0}{v^2}+\frac{16MI_0\log v}{v^3}+\mathcal{O}(v^{-3}) 
	\end{equation}
	for \eqref{eq:hot:1}
and, similarly, for \eqref{eq:hot:2}:
\begin{equation}
		\pv(r\phi)(0,r)=\frac{4J_0}{v^3}+\frac{48MJ_0\log v}{v^4}+\mathcal{O}(v^{-4}) .
	\end{equation}
These logarithmic higher-order asymptotics on initial data then lead (together with the contribution coming from the leading-order asymptotics of the solution itself) to logarithmic higher-order corrections in the asymptotic expansion of $\phi$ near $i^+$. 

Now, from the viewpoint of~\cite{Kerrburger}, the asymptotics \eqref{eq:hot:1}, \eqref{eq:hot:2} are of course inappropriate since they assume conformal smoothness or compact support: Instead, the results of~\cite{Kerrburger} motivate the following asymptotics on $u=0$, where $J_0$, $J_0'$ and $J_0''$ are constants:
	\begin{equation}\label{eq:hot:1log}
		\partial_r(r\phi)(0,r)=\frac{I^{\frac{\log r}{r^3}}_0[\phi]\log r}{r^3}+\frac{J_0}{r^3}+\mathcal{O}(\max(r^{-3-\epsilon_\phi},r^{-4}\log r)) 
	\end{equation}
	or
	\begin{equation}\label{eq:hot:2log}
		\partial_r(r\phi)(0,r)=\frac{J'_0}{r^3}+\frac{J''_0\log r}{r^4}+\mathcal{O}(r^{-4}) ,
	\end{equation}
where the latter expansion was obtained by e.g.\ considering the scattering problem on Schwarzschild with smooth compactly supported scattering data on $\mathcal{I}^-$ and $\mathcal{H}^-$ (and is also the generic case if the past Newman--Penrose constant vanishes and the data on $v=1$ are conformally smooth), see Theorems~1.5 or~6.2 in~\cite{Kerrburger}.

In the latter case \eqref{eq:hot:2log}, the situation is similar to \eqref{eq:hot:2}, with the difference being that, now, there are two contributions to the logarithmic corrections on data, so, in principle, there could be cancellations in the higher-order late-time asymptotics (this would depend on the extension of the data to $\mathcal H^+$):
	\begin{equation}\label{eq:hot:1logg}
		\partial_v(r\phi)(0,r)=\frac{4J'_0}{v^3}+\frac{(8J''_0+48MJ'_0)\log v}{v^4}+\mathcal{O}(v^{-4}).
	\end{equation}
On the other hand, in the former case \eqref{eq:hot:1log}, we have
\begin{equation}\label{eq:hot:2logg}
		\partial_v(r\phi)(0,r)=\frac{4I^{\frac{\log r}{r^3}}_0[\phi](\log v-\log 2)}{v^3}+\frac{4J_0}{v^3}+\frac{48MI^{\frac{\log r}{r^3}}_0\log^2 v}{v^4}+\mathcal{O}(\max(v^{-3-\epsilon_\phi},v^{-4}\log v)).
	\end{equation}
From this, one already sees that, if one were to suitably adapt the methods of $\cite{Angelopoulos2019LogarithmicInfinity}$, and if one assumes that $\epsilon_\phi\geq 1$, then one would find corrections to the asymptotics of $\phi$ near $i^+$  which contain $\log^2$-terms. (Again, there would also be a contribution coming from the leading-order asymptotics themselves.) For instance, for the radiation field along future null infinity, we would obtain a correction at order $u^{-3}\log^2 u$. We leave the details to the reader.

\section{Comments on the non-linear problem}\label{sec:nonlinearcomments}
We have shown in this paper that, in the case of the linear wave equation on a fixed Schwarzschild background, the logarithmic terms obtained near spacelike infinity in~\cite{Kerrburger} can be translated into leading order logarithmic asymptotics for the radiation field near $i^+$, $r\phi|_{\mathcal I^+}=Cu^{-2}\log u+\mathcal O(u^{-2})$. 
Now, the results of~\cite{Kerrburger} were, in fact, derived for the non-linear Einstein-Scalar field system and then obtained for the linear problem as a corollary.
It is therefore interesting to ask if one can show analogous results to the ones obtained here for the non-linear case.
We here only make two brief comments on two works which are related to this problem. 
\paragraph{The black hole case}
If we consider the Einstein-Scalar field system under spherical symmetry, with assumptions as in~\cite{Dafermos2005} (in particular, we assume that an event horizon exists), and the additional assumption that, on some outgoing null hypersurface, 
\begin{equation}|\pv(r\phi)|\leq C r^{-3}\log r
\end{equation}
then we are in the realm of the non-linear proof of (almost) Price's law by Dafermos--Rodnianski~\cite{Dafermos2005} (again with the exception of the logarithmic terms). Carefully following their arguments, one finds the following results (with a choice of coordinates as in~\cite{Dafermos2005}, see their Thm.~1.1): 
Along the event horizon, one obtains
\begin{equation}\label{Price1}
    |\phi|+|\pv\phi|\leq C_\epsilon v^{-3+\epsilon},
\end{equation}
whereas, along null infinity, one has
\begin{align}
    |r\phi|\leq Cu^{-2}\log u,&&
    |\pu(r\phi)|\leq C_\epsilon u^{-3+\epsilon},\label{Price3}
\end{align}
for any $\epsilon>0$ and a constant $C>0$. $C_\epsilon $ is a constant that blows up as $\epsilon\to 0$.
Note that these are only upper bounds. In particular, we expect that one can replace the $\epsilon$-growth in \eqref{Price1}, \eqref{Price3} with a logarithm.
\paragraph{The dispersive case}
If we consider the Einstein-Scalar field system under spherical symmetry and, instead of assuming that a black hole forms, assume that the solution disperses (and has a regular centre $r=0$), then, under the assumptions of~\cite{Luk2015QuantitativeSymmetry} and the additional assumption that 
\begin{equation}\pv(r\phi)\sim C (r+1)^{-3}\log (r+1)\end{equation}
on an outgoing null ray, we are in the setting of the proof of Luk--Oh~\cite{Luk2015QuantitativeSymmetry} (except for the $\log$-term).
A very slight adaptation of their methods then gives the following \textit{sharp} rates near timelike infinity: There exist positive constants $A,B$ such that, near $i^+$: 
\begin{align}
  A\min\{u^{-3}\log u,r^{-1}u^{-2}\log u\}&\leq  \phi\leq B\min\{u^{-3}\log u,r^{-1}u^{-2}\log u\},\\
  A \min\{u^{-3}\log u,r^{-3}\log r\}&\leq  -\pv(r\phi)\leq B \min\{u^{-3}\log u,r^{-3}\log r\},\\
   Au^{-3}\log u&\leq \pu(r\phi) \leq B u^{-3}\log u.
\end{align}
See also Thms.~3.1 and~3.16 in~\cite{Luk2015QuantitativeSymmetry}. To obtain the corresponding results for the logarithmic terms, one needs to slightly change the proof of Lemma~6.6 as well as the proof in section~10.
\section*{Acknowledgements}\addcontentsline{toc}{section}{Acknowledgements}
The author would like to thank Dejan Gajic for helpful discussions, and John Anderson, Mihalis Dafermos and Dejan Gajic for comments on drafts of the present work.

\appendix
\stoptocwriting
\newpage
\section{Definitions of the main energy norms}\label{app}

\subsection{The basic energy currents}
We define, with respect to any coordinate basis, the following energy momentum tensor for any scalar field $\phi$:
\begin{align*}
\mathbf{T}_{\mu\nu}[\phi]:=\partial_\mu \phi\partial_\nu \phi-\frac12 g_{\mu\nu} \partial^\xi\phi\partial_\xi\phi.
\end{align*}
Note that if $\phi$ is a solution to the wave equation, then $\mathbf{T}[\phi]$ is divergence-free.

For any vector field $V$, we further define the energy current $J^V[\phi]$ according to
\begin{align*}
J^V[\phi](\cdot):=\mathbf{T}[\phi](V,\cdot).
\end{align*}
\subsection{Definition of the energy norms}
We work in $(u,r,\theta, \varphi)$-coordinates, where $\partial_r=\frac{D}{2}L$. 

Let $\phi$ be a \underline{spherically symmetric} solution to $\Box_g\phi=0$ in the sense of Proposition~\ref{prop:Cauchy}, and let $\epsilon>0$.  Then we define the following initial data energy norms on $\Sigma_0$, where the subscript "${}_0$" of the energy norms below denotes the fact that these are the energy norms for the $\ell=0$-mode.
\begin{align}
E^N_{k}[\phi]:=\sum_{i\leq k }\int_{\Sigma_0}J^N[T^i\phi]\cdot n_{\Sigma_0}\dd \mu_{\Sigma_0},
\end{align}
\begin{nalign}
 E^\epsilon_{0,I_0\neq 0; k}&[\phi]:=E^N_{3+3k}[\phi]+\int_{N_0^\mathcal{I}} r^{3+2k-\epsilon}(\partial_r^{1+k}(r\phi))^2\dd r\\
+&\sum_{i\leq 2k}\int_{N_0^\mathcal{I}} r^{3-\epsilon}(\partial_r T^i(r\phi))^2+r^2(\partial_r T^{1+i}(r\phi))^2+r(\partial_r T^{2+i}(r\phi))^2\dd r\\
+&\sum_{\substack{m\leq k;\\  i\leq 2k-2m+\min(k,1)}}\int_{N_0^\mathcal{I}} r^{2+2m-\epsilon}(\partial_r^{1+m}T^i(r\phi))^2\dd r,
\end{nalign}
\begin{nalign}
 E^\epsilon_{0,I_0= 0; k}&[\phi]:=E^N_{5+3k}[\phi]+\int_{N_0^\mathcal{I}} r^{5+2k-\epsilon}(\partial_r^{1+k}(r\phi))^2\dd r\\
+&\sum_{i\leq 2k}\int_{N_0^\mathcal{I}} r^{5-\epsilon}(\partial_r T^i(r\phi))^2+r^{4-\epsilon}(\partial_r T^{1+i}(r\phi))^2+r^{3-\epsilon}(\partial_r T^{2+i}(r\phi))^2\\
&\phantom{\sum_{i\leq 2k}\int_{N_0^\mathcal{I}}}+ r^2(\partial_r T^{3+i}(r\phi))^2+r(\partial_r T^{4+i}(r\phi))^2\dd r\\
+&\sum_{\substack{m\leq k;\\  i\leq 2k-2m+\min(k,1)}}\int_{N_0^\mathcal{I}} r^{4+2m-\epsilon}(\partial_r^{1+m}T^i(r\phi))^2\dd r.
\end{nalign}
We further define 
\begin{nalign}
\widetilde{E}^\epsilon_{0,I_0\neq 0; k}&[\phi]:=E^\epsilon_{0,I_0\neq 0; k}[\phi]+\sum_{i\leq k }\int_{\Sigma_0}J^N[NT^i\phi]\cdot n_{\Sigma_0}\dd \mu_{\Sigma_0}
\end{nalign}
and
\begin{nalign}
\widetilde{E}^\epsilon_{0,I_0= 0; k}&[\phi]:=E^\epsilon_{0,I_0= 0; k}[\phi]+\sum_{i\leq k }\int_{\Sigma_0}J^N[NT^i\phi]\cdot n_{\Sigma_0}\dd \mu_{\Sigma_0}.
\end{nalign}
In an abuse of notation, we finally define
\begin{nalign}
E^\epsilon_{0,\ill\neq 0; k}[\phi]:=E^\epsilon_{0,I_0\neq 0; k}[\phi],&&\widetilde{E}^\epsilon_{0,\ill\neq 0; k}[\phi]:=\widetilde{E}^\epsilon_{0,I_0\neq 0; k}[\phi]
\end{nalign}
and 
\begin{nalign}
E^\epsilon_{0,\ill=0; k}[\phi]:=E^\epsilon_{0,I_0=0 ; k}[\phi],&&\widetilde{E}^\epsilon_{0,\ill=0; k}[\phi]:=\widetilde{E}^\epsilon_{0,I_0=0 ; k}[\phi].
\end{nalign}
This notation reflects the fact that the energy norms above "are blind" to logarithmic terms.
\resumetocwriting
{\small\bibliographystyle{ieeetr} 
\bibliography{references,references2,references3}}

\end{document}